\documentclass{article}

\usepackage{arxiv}

\usepackage[utf8]{inputenc} %
\usepackage[T1]{fontenc}    %
\usepackage{hyperref}       %
\usepackage{url}            %
\usepackage{booktabs}       %
\usepackage{amsfonts}       %
\usepackage{nicefrac}       %
\usepackage{microtype}      %
\usepackage{xcolor}         %

\usepackage{physics}
\usepackage{amsmath}
\usepackage{amsthm}
\usepackage{amssymb}
\usepackage{bm}
\usepackage{graphicx}
\usepackage{subfigure}

\newtheorem{theorem}{Theorem}[section]
\newtheorem{corollary}{Corollary}[section]
\newtheorem{lemma}{Lemma}[section]

\newtheorem{proposition}{Proposition}[section]

\def\>{\ensuremath{\rangle}}
\def\<{\ensuremath{\langle}}

\newcommand {\qbit}[2]{{|#1\>_{#2}\<#1|}}

\newcommand{\myproof}[1]{\textit{Proof}. #1 \hfill\qedsymbol}

\title{Towards Efficient Ansatz Architecture for Variational Quantum Algorithms}

\author{%
Anbang Wu \\
  Department of Computer Science\\
  University of California, Santa Barbara \\
  \texttt{anbang@ucsb.edu} \\
  \And
  Gushu Li \\
  Department of Electrical \& Computer Engineering\\
  University of California, Santa Barbara \\
  \texttt{gushuli@ece.ucsb.edu} \\
  \AND
  Yuke Wang \\
  Department of Computer Science\\
  University of California, Santa Barbara \\
  \texttt{yuke\_wang@cs.ucsb.edu} \\
  \AND
  Boyuan Feng \\
  Department of Computer Science\\
  University of California, Santa Barbara \\
  \texttt{boyuan@cs.ucsb.edu} \\
  \AND
  Yufei Ding \\
  Department of Computer Science\\
  University of California, Santa Barbara \\
  \texttt{yufeiding@cs.ucsb.edu} \\
  \AND
  Yuan Xie \\
  Department of Electrical \& Computer Engineering\\
  University of California, Santa Barbara \\
  \texttt{yuanxie@ucsb.edu} \\
}

\begin{document}

\maketitle

\begin{abstract}
Variational quantum algorithms are the quantum analogy of the successful neural network in classical machine learning.
The ansatz architecture design, which is similar to classical neural network architecture design,  is critical to the performance of a variational quantum algorithm.
In this paper, we explore how to design efficient ansatz and 
investigate several common design options in today's quantum software frameworks.
In particular, we study the number of effective parameters in different ansatz architectures and theoretically prove that an ansatz with parameterized RX and RZ gates and alternating two-qubit gate entanglement layers would be more efficient (i.e., more effective parameters per two-qubit gate). Such ansatzes are expected to have stronger expressive power and obtain better solutions. 
Numerical experimental results show that our efficient ansatz architecture outperforms other ansatzes with smaller ansatz size and better optimization results.

\end{abstract}

\section{Introduction}

Quantum machine learning~\cite{biamonte2017quantum} is an emerging machine learning paradigm due to its intrinsic large state space and the ability to fast generating probabilistic distributions that are justifiably hard to sample from a classical machine~\cite{aaronson2017complexity,boixo2018characterizing}.
Variational Quantum Algorithms (VQAs), which come with intrinsic noise resilience~\cite{JarrodRMcClean2016TheTO, Sharma2019NoiseRO} and modest computation resource requirement~\cite{cerezo2020variational}, are expected to tackle 
machine learning tasks and demonstrate practical usage on the near-term quantum computing devices.
VQA-based quantum machine learning has been applied on various tasks, including  optimization~\cite{Farhi2014AQA, Hadfield2019FromTQ}, chemistry simulation~\cite{McArdle2020QuantumCC, Peruzzo2014AVE}, quantum program compilation~\cite{Khatri2018QuantumAQ, Sharma2019NoiseRO}, etc.

VQA can be considered as the quantum version of the classical neural network (NN)~\cite{Killoran2018ContinuousvariableQN}.
The structure of the model in VQA, which is often termed as `ansatz' in the quantum computing community, is a parameterized quantum circuit\footnote{the conventional name of quantum program under the well-adopted quantum circuit model~\cite{nielsen2010quantum}} similar to the NN architecture.
In a VQA ansatz, there are some single-qubit gates with parameters applied on different qubits locally. There are also some two-qubit gates connecting different qubit pairs and formulating a network structure. 
The ansatz is executed on a quantum processor to evaluate a cost function determined by the specific machine learning task, which is similar to a forward execution of NN. 
The parameters in the ansatz are tuned through a classical optimizer to minimize the cost function, which is the analogy of %
NN training.

Similar to classical NN, %
 the ansatz architecture is critical to the performance of VQA and is actively being studied.
Some ansatz architectures, e.g., Quantum CNN~\cite{Cong2018QuantumCN}, Quantum RNN~\cite{Bausch2020RecurrentQN}, Quantum LSTM~\cite{Chen2020QuantumLS}, and Quantum GNN~\cite{Verdon2019QuantumGN}, are inspired by successful classical NN architectures.
Some other ansatz architectures are inspired by the application, e.g., the Unitary Coupled-Cluster Singles and Doubles (UCCSD) ansatz~\cite{Barkoutsos2018QuantumAF} for 
chemistry simulation %
and the Quantum Alternating Operator Ansatz~\cite{Farhi2014AQA} for combinatorial problems.
However, such ansatzes have \textit{not} yet been demonstrated in experiments due to their relatively high resource requirements and complex structures.

Except for the architectures mentioned above, one type of ansatz architecture of particular interest is the hardware-efficient ansatz (HEA) employed in several recent experimental demonstrations of VQAs~\cite{Arute2020HartreeFockOA,arute2020quantum,Kandala2017HardwareefficientVQ}.
HEA usually employs layers of parameterized single-qubit gates, which can provide linear transformation (similar to the weight matrix in NN), and layers of two-qubit gates, which can bring entanglement in the state and enrich the expressive power (like the non-linear layers in NN). 
The architectures of HEAs are regular and  easy to be efficiently mapped onto near-term hardware with limited computation capability.
In this paper, we focus on HEA which is probably the most practical model for QML in the next decade.

One drawback of HEA is that its solution space may not include the desired solution due to its application-independent design.
Today's quantum computing software frameworks like Qiskit~\cite{Qiskit}, provide multiple reconfigurable HEA templates while the best HEA design in practice is not known.
One active research direction to evaluate the quality of HEA designs is to study the expressive power of different HEA designs~\cite{Sim2019ExpressibilityAE,Rasmussen2020ReducingTA,Nakaji2020ExpressibilityOT, Funcke2020DimensionalEA}.
In general, a HEA with more gates and parameters is more expressive and more likely to capture the desired solution.

Recall that the original design objective of HEA is to reduce the gate count 
for efficiency on hardware.
Being more expressive with more gates is clearly in the opposite direction.
For example, the strong expressive power of quantum states comes from the entanglement among different qubits and the entanglement must be established by the non-local two-qubit gates.
But the two-qubit gates usually have relatively long latency and high error rates.
For example, the average error rate of two-qubit gates on IBM's superconducting quantum devices is $10\times$ higher than that of single-qubit gates.

In this paper, we investigate the efficient ansatz architecture design and our goal is to provide more expressive power with fewer gates and parameters.
We analyze several common parameterized single-qubit layer designs and two-qubit entanglement layer designs.
In particular, we study the number of effective parameters (parameters that cannot be combined with other ones) and the number of the expensive two-qubit gates in different architectures.
The efficiency can then be improved by increase the number of effective parameters per two-qubit gate.

Our theoretical analysis is based on a widely-used hardware model where the two-qubit gates are fixed CX gates and single-qubit gates are parameterized rotation gates along different axes.
We made several key observations.
\textbf{First}, if an ansatz consists of a single type of rotation (e.g., rotation along X-axis, RX) in the entanglement layer, the upper bound of the number of effective parameters is small.
\textbf{Second}, increasing the number of CX gates in the entanglement layer may not enrich the expressive power. 
\textbf{Third}, an ansatz made up of two types of rotation gates (e.g., rotation along X- and Y-axis, RX-RZ) and an alternating entanglement layer is more efficient. %
These observations provide insights about how to get more expressive power with fewer gates (especially two-qubit gates).

We evaluate our theoretical results by comparing the RX-RZ-CX ansatz with other common ansatz architectures on various VQA tasks, including molecule simulation and combinatorial problems.
Experiment results show that the simulation error of RX-RZ-CX ansatz with alternating entanglement is only around 4.4\% of that of the RX-CX ansatz with a linear CX entanglement on average for molecule simulation tasks. 
For combinatorial optimization tasks, all tested ansatz architectures can achieve the optimal solution while the RX-RZ-CX ansatz has the smallest number of CX gates.

Our major contributions can be summarized as follows:
\begin{enumerate}
\vspace{-8pt}
    \item We study the expressive power of commonly-used hardware-efficient ansatzes via calculating the number of effective parameters.
    \vspace{-3pt}
    \item We propose and examine various common HEA configuration options to help select efficient architectures for HEA. 
    \vspace{-3pt}
    \item Our numerical experiments demonstrate the effectiveness of the suggested ansatz design. The proposed RX-RZ-CX ansatz with alternating entanglement can achieve better solutions comparing with the RX-CX ansatz with linear entanglement.
\end{enumerate}

\subsection{Related Work}

The architecture of ansatz is one of the key design problems for VQAs~\cite{cerezo2020variational}.
Today's quantum software frameworks, e.g., IBM's Qiskit Aqua library~\cite{Qiskit}, provide various built-in ansatz architectures.
The design methodology of VQA ansatz can be classified into two types. 
The first type is to design the ansatz architecture based on the target problem, e.g., the UCCSD ansatz for the chemistry simulation applications~\cite{peruzzo2014variational}.
These ansatzes are in general easier to converge~\cite{Wiersema2020ExploringEA} while they require a large number of gates and are beyond the computation capability of noisy near-term devices.
The second type of ansatz is usually termed as hardware-efficient ansatz and its architecture is designed based on the target hardware platform~\cite{Kandala2017HardwareefficientVQ}.
These ansatzes have relatively lower gate counts and have been widely employed to experimentally demonstrate VQAs~\cite{Arute2020HartreeFockOA, arute2020quantum,  Kandala2017HardwareefficientVQ, kokail2019self}.

It is not yet fully resolved how to optimally construct a hardware-efficient ansatz.
One active research area towards this objective to is study the expressive power of different ansatz constructions
because an ansatz with a strong expressive power can represent more complicated functions~\cite{Hubregtsen2020EvaluationOP} and is more likely to cover the desired solution. 
Sim et al. proposed to analyze the ansatz expressive power from a statistical perspective~\cite{Sim2019ExpressibilityAE} and the proposed method is later employed by  \cite{Rasmussen2020ReducingTA} and \cite{Nakaji2020ExpressibilityOT} to detect redundant parameters and compare the expressive power of different ansatzes, respectively.
\cite{Funcke2020DimensionalEA} proposed a geometric algorithm that identifies superfluous parameters with ancilla qubits.
In contrast, this paper provides an algebraic analysis of commonly-used ansatzes. Our analysis directly identifies redundant parameters and suggests that the RX-RZ-CX ansatz with alternating entanglement is more efficient in the sense of the number of effective parameters per two-qubit gate.

It has also been noticed that highly expressive ansatzes are more difficult to train~\cite{Holmes2021ConnectingAE, Patti2020EntanglementDB, Tangpanitanon2020ExpressibilityAT}, which is the same with classical machine learning.
Various strategies~\cite{Rasmussen2020ReducingTA, Funcke2020DimensionalEA, Nakaji2020ExpressibilityOT, Grant2019AnIS} have been proposed to improve VQA trainability. 
How to efficiently train an expressive ansatz is worth further exploring yet beyond the scope of this paper.

\section{Preliminary}

\subsection{Quantum Computation Basics}

The quantum bit (qubit) is the basic information processing unit in quantum computing. 
The state of a qubit can be represented by a unit vector in a 2-dimensional Hilbert space spanned by the computational basis states $\ket{0}$ and $\ket{1}$ of the qubit.
Specifically, the state $\ket{\psi}$ of a qubit can be expressed by $\ket{\psi} = a\ket{0}+b\ket{1}$ where $a,b \in \mathbb{C}$ and $\abs{a}^2 + \abs{b}^2 = 1$.
Without ambiguity, $\ket{\psi}$ can be written in the vector form $\ket{\psi} = (a,b)^\intercal$ which is known as the state vector.
For a system with $n$ qubits, its state space is the tensor product of all its qubits and is a $2^n$-dimensional Hilbert space.
A state vector in this space is a vector with $2^n$ elements $\ket{\psi}=(a_0,\dots,a_{2^n-1})^\intercal$.

The most popular quantum computation model is the circuit model~\cite{nielsen2010quantum}.
In this model, quantum algorithms are represented by quantum circuits which are composed of sequences of quantum gates and measurements applied on qubits.
Quantum gates on an $n$-qubit system are unitary operators in the $2^n$-dimensional Hilbert space.
For example, the single-qubit bit-flip gate $\mathrm{X}$ can convert quantum state $\ket{0}$ to quantum state $\ket{1}$:
$
\mathrm{X} = \begin{pmatrix}
    0 & 1 \\
    1 & 0
\end{pmatrix}, \ \mathrm{X}\ket{0}=\begin{pmatrix}
    0 & 1 \\
    1 & 0
\end{pmatrix}\begin{pmatrix}
    1 \\
    0
\end{pmatrix} = \begin{pmatrix}
    0 \\
    1
\end{pmatrix} = \ket{1}.$

In this paper, we only consider single-qubit gates and two-qubit gates. Because they have been proved to be universal~\cite{Dawson2006TheSA} and most existing quantum hardware platforms only support single- and two-qubit gates.
We list the definitions of the $x$-axis rotation gate $\mathrm{RX}$, $z$-axis rotation gate $\mathrm{RZ}$, a two-qubit gate $\mathrm{CX}$, and their circuit representation in the following:
$$\begingroup
\setlength\arraycolsep{2pt}
\mathrm{RX}(\theta) =\begin{pmatrix}
\cos{\frac{\theta}{2}}   & -i\sin{\frac{\theta}{2}} \\
-i\sin{\frac{\theta}{2}} & \cos{\frac{\theta}{2}}
\end{pmatrix}, \mathrm{RZ}(\theta) = 
    \begin{pmatrix}
        e^{-i \frac{\theta}{2}}   & 0 \\
        0 & e^{i\frac{\theta}{2}}
    \end{pmatrix},
\mathrm{CX} = \begin{pmatrix}
    1 & 0 & 0 & 0 \\
    0 & 1 & 0 & 0 \\
    0 & 0 & 0 & 1 \\
    0 & 0 & 1 & 0 \\
\end{pmatrix}
\endgroup
,
\begin{array}{c}
\includegraphics[width=90pt]{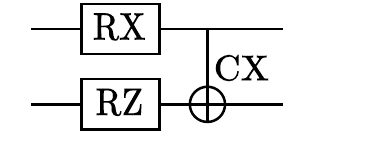}
\end{array}
$$
The RX and RZ gates are parameterized single-qubit gates. The CX is a fixed two-qubit gate.

\textbf{Entanglement:} The entanglement is a unique feature in quantum computing and can be created by two-qubit gates.
For example, CX gate can transform quantum state $\ket{\psi_0} = \frac{1}{\sqrt{2}}(\ket{0}+\ket{1})\ket{0}$ of qubits $q_0$ and $q_1$ to $\ket{\psi_1} = \frac{1}{\sqrt{2}}(\ket{00}+\ket{11})$:
$$\begin{pmatrix}
    1 & 0 & 0 & 0 \\
    0 & 1 & 0 & 0 \\
    0 & 0 & 0 & 1 \\
    0 & 0 & 1 & 0 \\
\end{pmatrix}\begin{pmatrix}
    \frac{1}{\sqrt{2}} \\
    0 \\
    \frac{1}{\sqrt{2}} \\
    0
\end{pmatrix} = \begin{pmatrix}
    \frac{1}{\sqrt{2}} \\
    0 \\
    0 \\
    \frac{1}{\sqrt{2}}
\end{pmatrix}.$$
The state $\ket{\psi_1}$ cannot be generated from $\ket{\psi_0}$ by only using single-qubit gates.
Such ability of a two-qubit gate to produce correlations between different qubits can greatly enrich the expressive power of quantum algorithms~\cite{Sim2019ExpressibilityAE}.

The measurements are not covered in this section since measurement usually does not appear in the HEA architecture, which is the focus of this paper. We refer the reader to~\cite{nielsen2010quantum} for details.

\subsection{VQA Basics and Ansatz Architecture}

The quantum circuit of a VQA is a parameterized circuit (a.k.a. ansatz) and some gates can be tuned to change their functions (e.g., the $\rm RX(\theta)$ gate mentioned above). 
Therefore, the final state $\ket{\psi(\bm{\theta})}$ after executing the ansatz will depend on the parameters. 
The solution of the VQA task is usually encoded in the minimal eigenvalue (or the corresponding eigenstate) of a Hermitian operator $O$ called observable.
On a quantum processor, we can evaluate $\bra{\psi(\bm{\theta})}O\ket{\psi(\bm{\theta})}$ and then optimize parameters $\bm{\theta}$ to minimize $\bra{\psi(\bm{\theta})}O\ket{\psi(\bm{\theta})}$.
After $\bra{\psi(\bm{\theta})}O\ket{\psi(\bm{\theta})}$ is  minimized, $\bra{\psi(\bm{\theta})}O\ket{\psi(\bm{\theta})}$ will be the minimal eigenvalue and $\ket{\psi(\bm{\theta})}$ is the corresponding eigenstate.

One natural question regarding the performance of VQA is that whether the eigenstate with respect to the smallest eigenvalue of the observable can be reached or approximated by tuning the parameters in the ansatz.
The answer is highly related to the architecture of the ansatz.
In this paper, we focus on one type of ansatz architecture, namely hardware-efficient ansatz (HEA), because it is constructed by gates that can be directly executed on the near-term quantum hardware.
A typical HEA has a layered architecture~\cite{Kandala2017HardwareefficientVQ} which consists of two-qubit gate layers to provide entanglement and single-qubit gate layers with tunable parameters.
Figure~\ref{fig:entanglementblock} (a) shows a linear entanglement layer with CX gates connecting all qubits one by one.
The single-qubit gate layer of hardware-efficient ansatz only contains rotation gates most time. Figure~\ref{fig:rxcxlayer} (a) provides an example of 
HEA which contains only single-qubit gates of RX and linear entanglement layer of CX gates. 
The design space of HEA can be explored in several directions, including changing the order of two-qubit gates in the entanglement layer (Figure~\ref{fig:entanglementblock} (b) (c)), applying two-qubit gates on different qubit pairs (e.g., alternating entanglement in Figure~\ref{fig:rxcxlayer} (b)), modifying the single-qubit gate layer (Figure~\ref{fig:rxcxlayer} (c)), etc.

\begin{figure}[t]
    \centering
    \includegraphics[width=0.5\columnwidth]{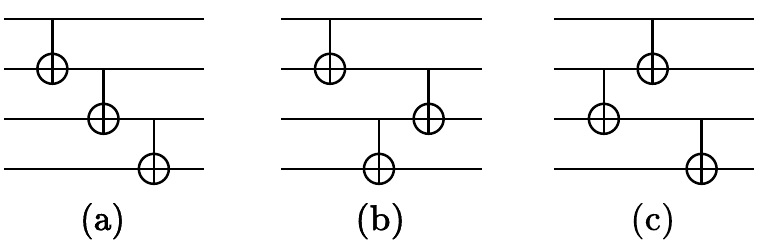}
    \vspace{-10pt}
    \caption{CX gates in an entanglement layer. (a) linear entanglement layer. (b)(c) examples of changing the order of CX gates in an entanglement layer}
    \label{fig:entanglementblock}
\end{figure}
\begin{figure}[t]
    \centering
     \vspace{-5pt}
    \includegraphics[width=\columnwidth]{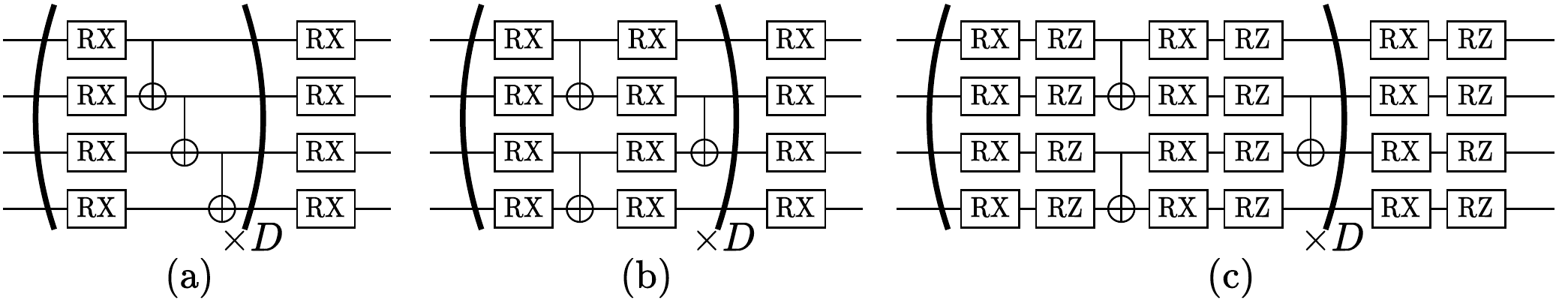}
    \vspace{-20pt}
    \caption{RX-CX ansatz with (a) linear entanglement (b) alternating entanglement, parameters in the RX gates are omitted %
    (c) RX-RZ-CX ansatz with alternating entanglement}
    \vspace{-10pt}
    \label{fig:rxcxlayer}
\end{figure}

\section{HEA Expressive Power Analysis}\label{sec:analysis}

In this section, we analyze the expressive power of commonly-used ansatz forms from existing quantum software frameworks.
We will focus on HEA since this type of ansatz has relatively low cost and is more likely to be deployed on near-term devices.
In particular, we study two key aspects in ansatz design, the two-qubit gate connection topology and the selection of single-qubit rotation gates.
We assume that all two-qubit gates are CX gates without parameters and all single-qubit gates are parameterized rotation gates. 
This is a widely used setting in prior research~\cite{Kandala2017HardwareefficientVQ, Sim2019ExpressibilityAE} and today's quantum software frameworks~\cite{Qiskit, Bergholm2018PennyLaneAD}.

\subsection{Effective Parameters of RX-CX Ansatz}
We first study the number of effective parameters since more effective parameters in an ansatz lead to stronger expressive power. 
We start from a simple case with one CX gate and two RX gates placed before and after the CX gate on the target qubit of the CX gate.. 
In the following proposition, we demonstrate an example in which the two parameters of the two RX gates can be aggregated.

\begin{proposition}[X-rule]\label{prop:x_rule}
The following two quantum circuits are equivalent. %
\begin{align*}
\vspace{-5pt}
     \includegraphics[height=5\fontcharht\font`\B]{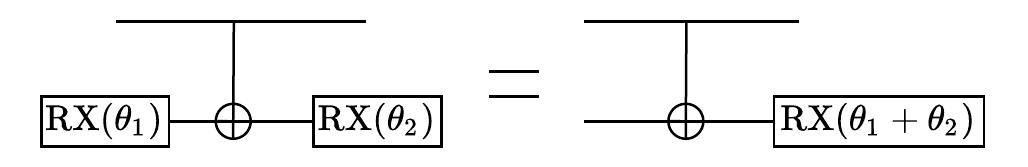} 
     \vspace{-10pt}
\end{align*}
\end{proposition}
\vspace{-5pt}
\myproof{The equivalence can be checked immediately by calculation.}

Since the two parameters $\theta_1$ and $\theta_2$ can be combined into one parameter $\theta = \theta_1+\theta_2$ in one RX gate in the proposition above, the circuit above with two parameters only has one effective parameter. 
This example suggests that an ansatz may have many parameters but the number of effective parameters may be smaller since some parameters can be combined.

We then study a more complicated case for RX-CX ansatz.
The RX-CX ansatz has a layered architecture (Figure~\ref{fig:rxcxlayer} (a)). 
In this ansatz, we first have one layer of RX gates on all qubits.
Then we employ an entanglement block, in which the qubits are connected by CX gates to create entanglement and enhance the expressive power of the ansatz.
We use linear entanglement in this case.
In the following proposition, we show that the number of effective parameters in an RX-CX ansatz with linear entanglement is very limited even if the ansatz is deep and has a large number of parameters.

\begin{proposition}\label{prop:rx_cx_upperbound}
Suppose the number of effective parameters of an n-qubit RX-CX ansatz with linear entanglement is $K$. We have the following conclusion:
$K\leq\lfloor\frac{3n^2+1}{4}\rfloor$
and the parameter combination is illustrated in Figure~\ref{fig:parametercombination}.
\end{proposition}
\myproof{Postponed to appendix~\ref{sec:rx-cx-eff}.}

This proposition provides an upper bound of the number of effective parameters in an RX-CX ansatz with linear entanglement. 
That is, even if such an ansatz has many layers and many single-qubit gates with parameters, the number of effective parameters will always be limited by this upper bound.
Surprisingly, the upper bound grows quadratically as the number of qubits increases. 
For an $n$-qubit system, the unitary transformation on it is a $2^n$-dimensional linear operator which can be considered as a matrix with $4^n$ elements.
Therefore, an ansatz would require an exponentially (with respect to the number of qubits) large number of effective parameters to become expressive and cover all possible unitary transformations.
However, an RX-CX ansatz with linear entanglement cannot support such a large number of effective parameters.
Thus, we argue that this ansatz architecture is not effective since it can only cover a small portion of all possible unitary transformations.

\begin{figure}[t]
    \centering
    \includegraphics[width=0.8\columnwidth]{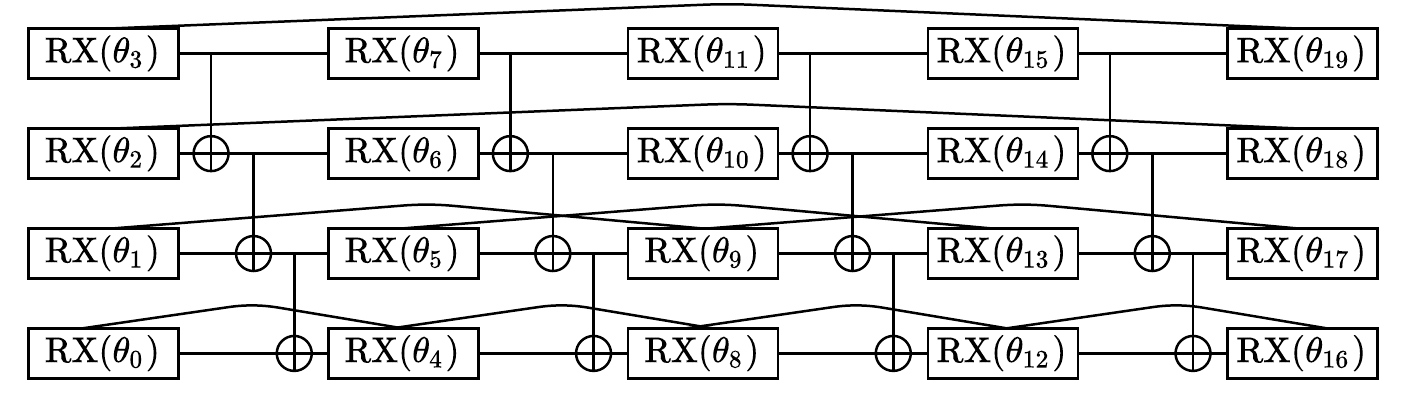}
    \vspace{-10pt}
    \caption{Parameter combination result for 4-qubit RX-CX ansatz with linear entanglement. Parameters that can be combined are connected by a curve.}
    \vspace{-10pt}
    \label{fig:parametercombination}
\end{figure}

\subsection{Order of the CX Gates}
We then investigate the variants of this RX-CX ansatz and we first modify the CX gates in the entanglement layer.
In the following proposition, we show that our result in Proposition~\ref{prop:rx_cx_upperbound} holds when changing the order of the CX gates in a linear entanglement layer.

\begin{proposition}[CX order does not matter]\label{prop:rxcx_cnotorder}
An n-qubit RX-CX ansatz with modified linear entanglement
has at most $\lfloor\frac{3n^2+1}{4}\rfloor $ effective parameters. The modified linear entanglement has alternated CX order compared with the original linear entanglement as shown in Figure~\ref{fig:entanglementblock} (b)(c).
\end{proposition}
\myproof{Postponed to appendix~\ref{sec:order_cnot}.}

Another option to change the CX gates in the entanglement layer is to employ the so-called  alternating entanglement layer. 
The architecture of RX-CX ansatz with alternating entanglement layer is shown in Figure~\ref{fig:rxcxlayer} (b).
In such architecture, we first have one layer of RX gates followed by one layer of CX gates applied on qubit pairs $(0, 1), (2, 3), \dots$.
We then have another layer of RX gates followed by one layer of CX gates applied on qubit pairs $(1, 2), (3, 4), \dots$.
The CX gates are applied to non-overlapping qubit pairs in one layer.

Comparing to the RX-CX ansatzes with linear entanglement, the RX-CX ansatzes with alternating entanglement have more single-qubit gates (or equivalently, more parameters) per CX gate.
Such property would be desirable if all parameters are effective parameters.
Because the CX gate has relatively high execution overhead (high error rate and long gate latency) and an efficient ansatz should  support more effective parameters with fewer CX gates. 
However, a direct corollary of Proposition~\ref{prop:rxcx_cnotorder} is that employing the alternating entanglement layer in the RX-CX ansatz \textit{cannot} improve the upper bound of the number of effective parameters.

\begin{corollary}\label{coro:rxcx-AL-layer}
A $(2L)$-layer RX-CX ansatz with alternating entanglement can be reduced to an $L$-layer RX-CX ansatz with linear entanglement.
\end{corollary}
\myproof{Postponed to appendix~\ref{sec:order_cnot}.}

\subsection{More CX Gates in the Entanglement Layer}

Another direction to change the ansatz architecture is to have more CX gates in the entanglement layer.
We propose the next two propositions to study the effect of adding CX gates in the entanglement layer. 
The construction of entanglement layer, even with CX gates only, is very complex because we can have many CX gates on arbitrary two qubits in arbitrary order.
Here we explore two approaches to increase the number of CX gates in the entanglement layer.

One direct change is to repeat the entanglement layer multiple times.
However, since the entanglement layer only has CX gates, repeating such a layer may not yield a high expressive power.

\begin{proposition}[Repetition of entanglement layer]\label{prop:rep_deg}
For any entanglement layer that is purely constructed by CX gates, there exists an integer $k$, s.t. $E^k = I$ where $E$ is the unitary transformation represented by the entanglement layer and $I$ is the identity operator.
\end{proposition}
\myproof{Postponed to appendix~\ref{sec:more_cnot}.}

This result suggests
that repeating an entanglement layer can only yield a finite number of different unitary transformations.
Actually, for any non-negative integer $t = sk + r$, $E^t = (E^k)^sE^r = I^s E^{r} = E^r$ by Proposition~\ref{prop:rep_deg}. Thus, the set of all possible numbers of repetitions is equivalent to the finite set $\{ I, E, E^2, \cdots, E^{k-1} \}$ which forms a cyclic group of order $k$ under matrix multiplication.
Therefore, it may not be efficient to repeat the entanglement layer since it will significantly increase the number of  expensive CX gates.

We also find that the functionality of the CX-based entanglement layer is limited and we formulate this statement in the next proposition.
\begin{proposition}[Limitation of CX-based entanglement layer]\label{prop:limit_cx}
For a given n-qubit unitary transformation $U$, $\forall i, j \ge 2$, it is possible to move $i$-th column of $U$ to $j$-th column of $U$ with a finite number of CXs, but it is not possible to achieve column swap between two arbitrary columns when $n \ge 3$. %
\end{proposition}
\myproof{Postponed to appendix~\ref{sec:more_cnot}.}

An entanglement layer with CX gates is not universal and it cannot even implement swapping two columns.
As a result, to increase the number of effective parameters per CX gate, it would be better to have more parameterized single-qubit gates between the entanglement layers.

\subsection{An Efficient Ansatz Architecture}
Finally, our objective is to have an efficient ansatz architecture design.
We are not going to give the optimal solution since it is a highly complicated problem and the optimal ansatz might depend on the optimization objective function.
Instead, we aim to give a better ansatz architecture design based on the propositions above.

We find that an ansatz with two types of single-qubit gates, RX and RZ, and alternating entanglement would be more efficient as it can give a higher number of effective parameters per CX gate. 
The architecture of this RX-RZ-CX ansatz is shown in Figure~\ref{fig:rxcxlayer} (c).
The alternating entanglement layers are the same as those in the RX-CX ansatz. 
The only difference is that we have two single-qubit gates RX and RZ on each qubit in each single-qubit layer.
We believe that such RX-RZ-CX ansatz with alternating entanglement will be efficient since the RZ gate provides rotations along the $z$-axis which is beyond the rotations along the $x$-axis provided by the RX gates and CX gates.
Also, the alternating entanglement employs relatively few CX gates to support one single-qubit layer.

\begin{theorem}[Efficiency of the RX-RZ-CX ansatz with alternating entanglement]\label{theo:rxrzcxefficiency}
For an $n$-qubit $2L$-layer RX-RZ-CX ansatz with alternating entanglement, the numbers of effective parameters (w.r.t parameter combination) $w_{xz}$ satisfies
$w_{xz} = (4n-3)*L$.
\end{theorem}
\myproof{Postponed to appendix~\ref{sec:rx-rz-cx-eff}.}

\begin{corollary}[] Theorem~\ref{theo:rxrzcxefficiency} can be generalized to other types of HEA such as RY-CX, RX-RY-CX and RY-RZ-CX. 
The number of effective parameters of an $n$-qubit $2L$-layer RY-CX, RX-RY-CX, or RY-RZ-CX ansatz with alternating entanglement is $w_{y} = (n-1)*2L$, $w_{yz} = (4n-3)*L$, or $w_{xy} = (4n-3)*L$, respectively.
\end{corollary}
\myproof{Postponed to appendix~\ref{sec:rx-rz-cx-eff}.}

The results in Theorem~\ref{theo:rxrzcxefficiency} shows that the number of effective parameters per CX gate of an RX-RZ-CX ansatz with alternating entanglement is far larger than that of an RX-CX ansatz with linear entanglement if they have the same number of layers.
As a result, the RX-RZ-CX ansatzes are expected to outperform the RX-CX ansatzes in VQA applications because the relative expressive power is higher and better solutions can possibly be covered by the expressive ansatzes.

\section{Evaluation}

In this section, we evaluate our major results in the last section by comparing the RX-RZ-CX ansatz with alternating entanglement and other ansatz architectures on various VQA benchmarks.

\subsection{Experimental Setup}
\textbf{Benchmarks:}
We select seven VQA benchmarks. The first four benchmarks are the ground state energy simulation of four molecules (${\rm H_2}$, ${\rm LiH}$, ${\rm BeH_2}$, ${\rm HF}$) using the variational quantum eigensolver (VQE)~\cite{Peruzzo2014AVE}.
The second three benchmarks are to solve three combinatorial optimization problems, the graph max-cut problem (MC), the vertex cover problem (VC), and the traveling salesman problem (TSP) using the quantum approximate optimization algorithm (QAOA)~\cite{Farhi2014AQA}.
The number of qubits and the minimum (optimal value) of the cost function of the benchmarks are listed in Table~\ref{tab:benchmark}.%

\begin{table}[b]
\caption{Benchmark information}
\resizebox{0.98\columnwidth}{!}{
\begin{tabular}{|c|c|c|c|c|c|c|c|}
\hline
Name             & VQE-${\rm H_2}$ & VQE-${\rm LiH}$ & VQE-${\rm BeH_2}$ & VQE-${\rm HF}$ & QAOA-MC & QAOA-VC & QAOA-TSP \\ \hline
\# of Qubits     & 4               & 6               & 8                 & 10             & 4       & 6       & 9        \\ \hline
optimal cost $E$ & -1.1373         & -7.8810         & -15.048           & -98.593        & 2.0     & 3.0     & 202.0    \\ \hline
\end{tabular}
}
\vspace{-10pt}
\label{tab:benchmark}%
\end{table}

\textbf{Ansatz architectures:} We use RX-CX ansatz with linear entanglement, denoted by `RX-CX-L', RX-CX ansatz with alternating entanglement, denoted by `RX-CX-A', and RX-RZ-CX ansatz with alternating entanglement, denoted by `RX-RZ-CX-A' , to illustrate the importance of effective parameters and the effect of the alternating entanglement. 
These ansatzes are either generated by Qiskit or obtained from previous work~\cite{cerezo2020cost, Nakaji2020ExpressibilityOT}.

\textbf{Metric:} We evaluate the approximation error to compare the performance of different ansatzes. The approximation error is defined as:
    $\epsilon = \abs{E_a - E}$,
where $E_a$ is the cost function value optimized by VQA 
, and $E$ is the minimum cost function value 
which is the smallest eigenvalue of the target observable. 
In our simulation-based experiments, $E$ can be obtained by applying the eigensolver from Numpy~\cite{Harris_2020} on the observable of the task to evaluate the performance of the ansatz architectures. 
For a practical problem that cannot be simulated classically, $E$ can be evaluated or cross-validated with other approaches, e.g., experiment results for chemistry simulation tasks.

\textbf{Implementation:} All experiments are implemented on IBM's Qiskit framework~\cite{Qiskit}. For molecule simulation problems (i.e., VQE tasks), we use the Qiskit Chemistry module and PySCF~\cite{PYSCF} driver to generate the Hamiltonian of target molecules. For combinatorial optimization tasks, we use the Qiskit Optimization module to generate the Hamiltonian of the simulated combinatorial optimization problems. 
We use the COBYLA optimizer~\cite{bos2006numerical} to optimize the parameters. 
Since today's public quantum devices cannot accommodate a full VQA training process, all experiments in this paper are performed in the Qiskit Simulator running on a server with a 6-core Intel E5-2603v4 CPU and 32GB of RAM.
Parameters are randomly initialized.

\subsection{Results}

Table~\ref{tab:result} shows the simulation result of different ansatzes on various VQA tasks. 
In general, the RX-RZ-CX-A ansatzes outperform other selected ansatz architectures with more accurate solutions for VQE chemistry simulation tasks. 
For QAOA tasks, the three types of ansatzes have comparable accuracy (all errors are smaller than $10^{-4}$). This is because the solution state of QAOA tasks is a computational basis state, which means a highly entangled final state may not be necessary and a high-quality result may be generated from simpler ansatzes.

\textbf{Effect of effective parameters:} We first compare RX-CX-L ansatz and RX-RZ-CX-A ansatz. With the same number of layers, the number of CX gates in an RX-RZ-CX-A ansatz is half of that in  an RX-CX-L ansatz. 
By Theorem~\ref{theo:rxrzcxefficiency}, RX-RZ-CX-A ansatz has more effective parameters than the RX-CX-L ansatz.
As expected, RX-RZ-CX-A ansatz can have much more accurate results than RX-CX-L ansatz and the error is reduced by $95.5\%$ for VQE tasks. 

We then compare the RX-CX-L ansatz and RX-CX-A ansatz, both of which only have parameters in RX gates.
When they share the same number of layers, we know that the RX-CX-L ansatz has twice as many effective parameters as the RX-CX-A ansatz by Corollary~\ref{coro:rxcx-AL-layer}.
The approximation error of the RX-CX-L ansatz is reduced by $37.5\%$ comparing with the error of the RX-CX-A ansatz (with the same number of layers) on VQE tasks.

These results indicate that incorporating
more effective parameters in the ansatz can improve the expressive power, and thus reduce the approximation error of the ansatz. 
There are mainly two ways to increase the number of effective parameters without changing the ansatz structure: introducing extra single-qubit gates or having more layers.
However, for the RX-CX ansatz, it is not possible to boost the number of effective parameters by adding more layers, as indicated by Proposition~\ref{prop:rx_cx_upperbound}.

\begin{table*}[t] \label{tab:expres}
  \centering
  \caption{Simulation results for VQE and QAOA tasks.
  }
    \resizebox{0.86\textwidth}{!}{
    \begin{tabular}{|c|c|c|c|c|c|c|}
    \hline
    name & ansatz  & \multicolumn{1}{c|}{\# of params.} & \multicolumn{1}{c|}{\# of CXs} & \multicolumn{1}{c|}{\# of layers} & \multicolumn{1}{c|}{ optimized cost $E_a$ } & \multicolumn{1}{c|}{ error $\epsilon$} \\
    \hline
    & RX-CX-L  & 20 & 12 & 4 & -0.83096 & 0.30634  \\
    \cline{2-7}
    VQE-${\rm H_2}$ & RX-CX-A & 20 & 6 & 4 & -0.82081 & 0.31649  \\
    \cline{3-7}
    &   & 40 & 12 & 8 & -0.83096 & 0.30634 \\
    \cline{2-7}
     & RX-RZ-CX-A  & 40 & 6 & 4 &  -1.1170 & \textbf{0.02030}  \\
    \hline
    & RX-CX-L  & 36 & 25 & 5 & -7.6532 & 0.22780  \\
    
    \cline{2-7}
    VQE-${\rm LiH}$  & RX-CX-A & 36 & 13 & 5 & -7.5202 & 0.36080 \\
    \cline{3-7}
    &   & 72 & 25 & 10 & -7.7222 & 0.15880 \\
    \cline{2-7}
    & RX-RZ-CX-A  & 72 & 13 & 5 & -7.8656 & \textbf{0.01540} \\
    \hline
    & RX-CX-L  & 48 & 35 & 5 & -14.802 & 0.24600 \\
    
    \cline{2-7}
    VQE-${\rm BeH_2}$ & RX-CX-A & 48 & 18 & 5 & -14.566 & 0.48200  \\
    \cline{3-7}
    & & 96 & 35 & 10 & -14.880 & 0.16800 \\
    \cline{2-7}
     & RX-RZ-CX-A  & 96 & 18 & 5 & -15.040 &  \textbf{0.00800} \\
    \hline
    & RX-CX-L  & 60 & 45 & 5 & -97.687 & 0.90600  \\
    
    \cline{2-7}
    VQE-${\rm HF}$ & RX-CX-A & 60 & 23 & 5 & -96.756 & 1.8370  \\
    \cline{3-7}
    &  & 120 & 45 & 10 & -97.969 & 0.62400 \\
    \cline{2-7}
     & RX-RZ-CX-A  & 120 & 23 & 5 & -98.568 & \textbf{0.02500} \\
    \hline
    & RX-CX-L  & 20 & 12 & 4 & 2.0000 & $<$ $10^{-4}$  \\
    
    \cline{2-7}
    QAOA-MC & RX-CX-A & 20 & 6 & 4 & 2.0000 & $<$ $10^{-4}$ \\
    \cline{3-7}
    &  & 40 & 12 & 8 & 2.0000 & $<$ $10^{-4}$\\
    \cline{2-7}
     & RX-RZ-CX-A  & 20 & 6 & 4 & 2.0000 & $<$ $10^{-4}$  \\
    \hline
    & RX-CX-L  & 36 & 25 & 5 & 3.0000 & $<$ $10^{-4}$ \\
    
    \cline{2-7}
    QAOA-VC  & RX-CX-A  & 36 & 13 & 5 & 3.0000 & $<$ $10^{-4}$ \\
    \cline{3-7}
    &  & 72 & 25 & 10 & 3.0000 & $<$ $10^{-4}$\\
    \cline{2-7}
    & RX-RZ-CX-A  & 72 & 25 & 5 & 3.0000 & $<$ $10^{-4}$ \\
    \hline
    & RX-CX-L  & 54  & 40 & 5 & 202.00 & $<$ $10^{-4}$ \\
    \cline{2-7}
    QAOA-TSP  & RX-CX-A  & 54 & 20 & 5 & 202.00 & $<$ $10^{-4}$\\
    \cline{3-7}
    & & 108 & 40 & 10 & 202.00 & $<$ $10^{-4}$ \\
    \cline{2-7}
    & RX-RZ-CX-A  & 108 & 20 & 5 & 202.00 & $<$ $10^{-4}$  \\
    \hline
    \end{tabular}%
    }
  \label{tab:result}%
\end{table*}

\textbf{Effect of entanglement layer architecture:} Besides the number of effective parameters, the structure of the entanglement layer may also affect the performance of the ansatz. With twice the number of layers, the RX-CX-A ansatz has the same number of effective parameters as the RX-CX-L ansatz by Corollary~\ref{coro:rxcx-AL-layer}. This indicates that the RX-CX-A ansatz needs to have more single-qubit gates and parameters than the RX-CX-L ansatz if we hope to keep the same number of effective parameters in the RX-CX-A and the RX-CX-L ansatz. 
In this case, many parameters in the RX-CX-A ansatz will be redundant.
In general, the ansatz with a larger number of redundant parameters has lower performance because of the worse trainability. 
However, in our experiment, the results of the RX-CX-A ansatz are slightly better than that of the RX-CX-L ansatz (error reduced by $22\%$ on average).
This may come from the better trainability of an ansatz with alternating entanglement, as indicated by~\cite{Nakaji2020ExpressibilityOT}.

\section{Discussion}\label{sec:discuss}
The ansatz architecture %
of VQA can significantly affect the optimization result quality, trainability, execution overhead, etc. 
An optimal ansatz architecture towards all these design objectives is, to the best of our knowledge, not yet known or may not exist.
In this paper, we focused on the \textit{efficiency} of different ansatz architectures and studied how to provide more expressive power with fewer gates (especially two-qubit gates).
We believe that this is critical if we hope to demonstrate some practical usage of VQA on near-term noisy quantum computing devices.
We investigated several common ansatz design options and found that using the RX-RZ-CX ansatz with alternating entanglement layer can provide more effective parameters per CX gate compared with other common architectures discussed in this study.
These results can guide the design of future VQA ansatz and, hopefully, make progress towards quantum advantage.

Although our conclusion has been evidenced by both theoretical analysis and experiment results, there is much work left since the design space is extremely large for VQA ansatz and there are multiple ansatz architecture optimization objectives.
Here we briefly discuss some potential future research directions.

\textbf{Application-specific HEA:} In this paper, we focus on efficient ansatz design and try to maintain strong expressive power with fewer CX gates for a smaller execution overhead.
The trainability of HEA is another critical problem of great interest since generally HEA trades in trainability for efficiency. 
The application-specific ansatzes are usually easier to train.
It is worth exploring if trainability and efficiency can be achieved simultaneously if application information is introduced into HEA architecture design.

\textbf{Parameter pruning:} Our conclusion suggests that deploying single-qubit gates of different types would increase the number of effective parameters. 
However, similar to classical NN, ansatz with a large number of effective parameters will be inevitably harder to train. 
One possible solution is to remove effective parameters that have relatively small effects on ansatz performance.
The parameter pruning techniques in classical NN may be migrated to VQA to reduce the total number of parameters and improve the trainability.

\textbf{More efficient entanglement:} Our theoretical analysis and experiment results have verified that the RX-RZ-CX ansatz with alternating entanglement can give a high-quality solution with relatively few CX gates.
However, the alternating entanglement layer may not be optimal and the efficiency of the entanglement layer can be further improved.
For example, we may not need to fully connect all qubits in a linear entanglement layer or connect completely different qubit pairs in successive entanglement layers.
Also, the CX gate may not be the most efficient two-qubit gate to create entanglement and we may investigate other common two-qubit gates like the $\rm \sqrt{iSWAP}$ and ZZ gate~\cite{Gokhale2020OptimizedQC}. %

\section{Conclusion}

In this paper, we provide an algebraic viewpoint of the expressive power of ansatz. 
By analyzing the number of effective parameters, we show that the RX-CX ansatz with linear entanglement has limited expressive power. We also examine several common design options that can possibly improve the expressive power of ansatz. 
Based on our theoretical analysis, we suggest that the RX-RZ-CX ansatz with alternating entanglement is more efficient than the RX-CX ansatz with linear entanglement in terms of the number of effective parameters per two-qubit gate. 
Experiment results support our theoretical claim. 
The theoretical framework proposed by this paper can be further extended to more complex ansatz design problems and pave the way to demonstrating near-term quantum advantage using VQAs.

\bibliography{example_paper}
\bibliographystyle{plain}

\newpage
\appendix

\section{Proof of the Propositions and Theorems in Section~\ref{sec:analysis}}

In the following, the qubit indices start from 0.

\subsection{Proposition~\ref{prop:rx_cx_upperbound}: Effective Parameters of RX-CX Ansatz}\label{sec:rx-cx-eff}

\begin{lemma} \label{lemma:xpass}
$\mathrm{CX} \cdot (\mathrm{X} \otimes \mathrm{I}) = (\mathrm{X} \otimes \mathrm{X}) \cdot \mathrm{CX}  = (\mathrm{X} \otimes \mathrm{I}) \cdot \mathrm{CX} \cdot (\mathrm{I} \otimes \mathrm{X}) $.
\begin{equation*}
    \includegraphics[height=8\fontcharht\font`\B]{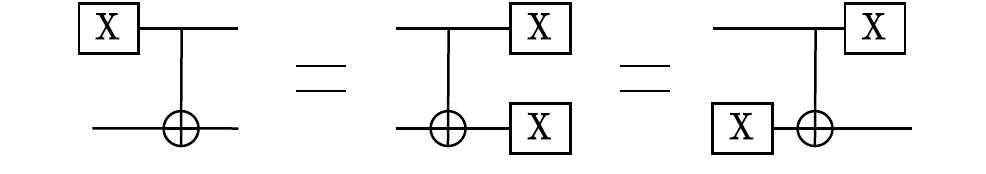} 
\end{equation*}
\end{lemma}
\begin{proof}
$\mathrm{CX} \cdot (\mathrm{X} \otimes \mathrm{I}) = \begin{pmatrix}
\mathrm{I} & \bm{0} \\
\bm{0} & \mathrm{X} 
\end{pmatrix}\begin{pmatrix}
\bm{0} & \mathrm{I} \\
\mathrm{I} & \bm{0} 
\end{pmatrix} = \begin{pmatrix}
\bm{0} & \mathrm{I} \\
\mathrm{X} & \bm{0} 
\end{pmatrix} = \begin{pmatrix}
\bm{0} & \mathrm{X} \\
\mathrm{X} & \bm{0} 
\end{pmatrix}\begin{pmatrix}
\mathrm{I} & \bm{0} \\
\bm{0} & \mathrm{X} 
\end{pmatrix} = (\mathrm{X} \otimes \mathrm{X}) \cdot \mathrm{CX}$. On the other hand, with Proposition~\ref{prop:x_rule}, we have $(\mathrm{X} \otimes \mathrm{X}) \cdot \mathrm{CX} = (\mathrm{X} \otimes \mathrm{I}) \cdot \mathrm{CX} \cdot (\mathrm{I} \otimes \mathrm{X})$.
\end{proof}

\begin{figure*}[h]
    \centering
    \includegraphics[width=1\textwidth]{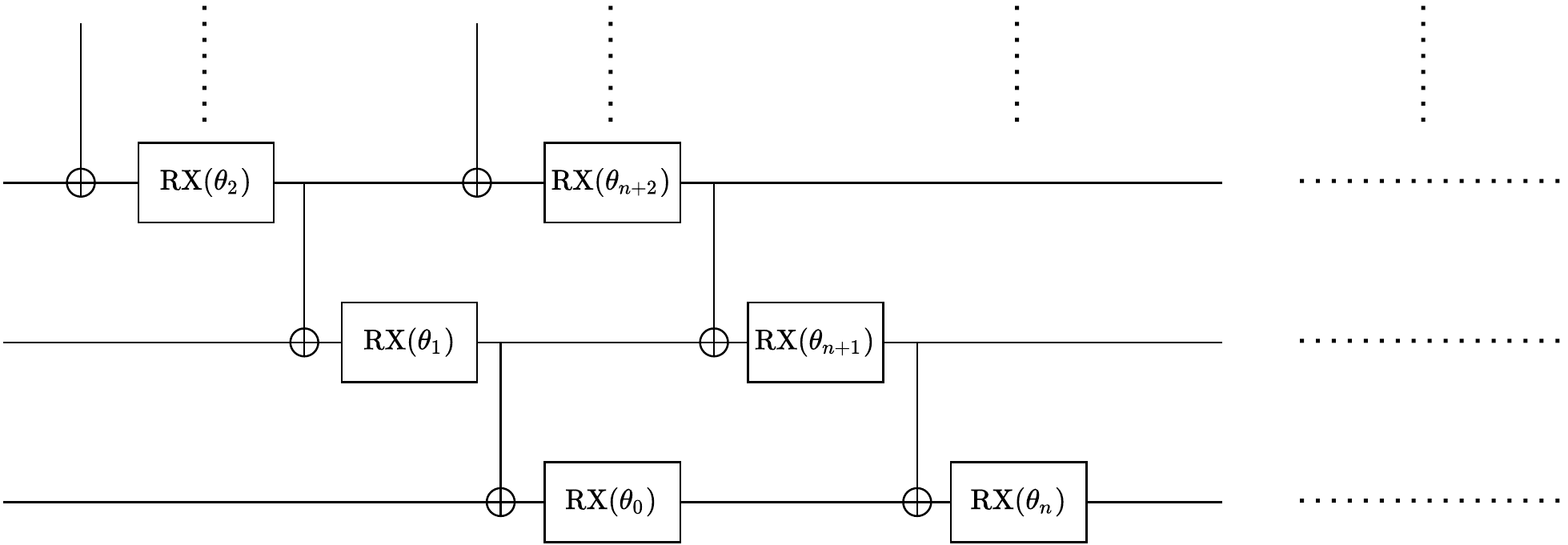}
    \caption{Reorganized RX-CX ansatz with linear entanglement by using Proposition~\ref{prop:x_rule}}
    \label{fig:intev}
\end{figure*}

\begin{lemma}(Layer unitary of RX-CX ansatz with linear entanglement) \label{lemma:layer_uni}
Let $U^0(\theta_{0}) = \mathrm{RX}(\theta_{0})$. Then the unitary for the first layer of $(k+1)$-qubit RX-CX ansatz with linear entanglement can be expressed as:
\begin{equation}\label{equ:recur_layer}
 U^k(\bm{\theta}_k) = \begin{pmatrix}
C_{k}U^{k-1}(\bm{\theta}_{k-1}) & -iS_{k}U^{k-1}(\bm{\theta}_{k-1}) \\
-iS_{k}U^{k-1}(\bm{\theta}_{k-1})X_{k-1} & C_{k}U^{k-1}(\bm{\theta}_{k-1})X_{k-1}  
\end{pmatrix}
\end{equation}
where $X_{k-1} =\begin{pmatrix}
0 & \mathrm{I}_{2^{k-1}} \\
\mathrm{I}_{2^{k-1}} & 0
\end{pmatrix}$, $C_{k} = \cos (\frac{\theta_{k}}{2})$, $S_{k} = \sin (\frac{\theta_{k}}{2})$, $\bm{\theta}_k = (\theta_k,\theta_{k-1},\cdots, 0)$. \\
For simplicity, in the following, we denote $U^k(\bm{\theta}_{k+n*j})$ by $U^k_j$.

\end{lemma}
\begin{proof}
The first layer of $(k+1)$-qubit RX-CX ansatz with linear entanglement is $(\mathrm{CX}(1,0)\cdots\mathrm{CX}(k-1,k))(\mathrm{RX}(\theta_{k})\otimes \mathrm{RX}(\theta_{k-1})\otimes\cdots\otimes \mathrm{RX}(\theta_{0}))$. 
Since CX gate and RX gate will exchange on the target qubit, we can convert the unitary to $\mathrm{RX}({\theta_0})(\mathrm{CX}(1,0)\mathrm{RX}({\theta_1}))(\mathrm{CX}(2,1)\mathrm{RX}({\theta_2}))\cdots(\mathrm{CX}(k,k-1)\mathrm{RX}({\theta_{k}}))$, as shown in Figure~\ref{fig:intev}.

Let $U^{k}(\bm{\theta}_k) = \mathrm{RX}({\theta_0})(\mathrm{CX}(1,0)\mathrm{RX}({\theta_1}))(\mathrm{CX}(2,1)\mathrm{RX}({\theta_2}))\cdots(\mathrm{CX}(k,k-1)\mathrm{RX}({\theta_{k}}))$, then we have $U^{k}(\bm{\theta}_k) = \mathrm{I}\otimes U^{k-1}(\bm{\theta}_{k-1})\mathrm{CX}(k,k-1)\mathrm{RX}({\theta_{k}})\otimes \mathrm{I} = \begin{pmatrix}
U^{k-1}(\bm{\theta}_{k-1}) & \bm{0} \\
\bm{0} & U^{k-1}(\bm{\theta}_{k-1})
\end{pmatrix}\begin{pmatrix}
\mathrm{I} & \bm{0} \\
\bm{0} & X_{k-1}
\end{pmatrix}\begin{pmatrix}
C_k\mathrm{I} & -iS_k \mathrm{I} \\
-iS_k \mathrm{I} & C_k\mathrm{I}
\end{pmatrix} = \begin{pmatrix}
C_{k}U^{k-1}(\bm{\theta}_{k-1}) & -iS_{k}U^{k-1}(\bm{\theta}_{k-1}) \\
-iS_{k}U^{k-1}(\bm{\theta}_{k-1})X_{k-1} & C_{k}U^{k-1}(\bm{\theta}_{k-1})X_{k-1}  
\end{pmatrix}$.
\end{proof}

\begin{lemma}\label{lemma:xcomp}
For $n$-qubit RX-CX ansatz with linear entanglement, given vector $\bm{b} \in \{0, 1\}^l$, Let $G^{k}(\bm{b}) = (U^{k}_{l-1}X_{k}^{b_{l-1}})(U^{k}_{l-2}X_{k}^{b_{l-2}})\cdots (U^{k}_{0}X_{k}^{b_{0}})$. Then, \\
\begin{equation*}
    \forall k \ge 0, l = 2^{\lceil \log_2 (k+2) \rceil}, \forall \bm{b} \in \{0, 1\}^{l}, G^{k}(\bm{b}) = G^{k}(\bar{\bm{b}}), 
\end{equation*}
where $\bar{\bm{b}} = (1-b_0,1-b_1,\cdots)$.
\end{lemma}
\begin{proof}
We prove that, $\forall k \in (2^{r}-2, 2^{r+1}-2]$, $G^{k}(\bm{b}) = G^{k}(\bar{\bm{b}})$. Note that, $\forall k \in (2^{r}-2, 2^{r+1}-2]$, we have $2^{\lceil \log_2 (k+2) \rceil} = 2^{r + 1}$.

It is easy to verify that $G^k(\bm{b}) = U^{k}_{l-1}(\theta_{k+n*(l-1)}+b_{l-1}*\pi,\theta_{k-1+n*(l-1)},\cdots)U^{k}_{l-2}(\theta_{k+n*(l-2)}+b_{l-2}*\pi,\theta_{k-1+n*(l-2)},\cdots)\cdots U^{k}_{0}(\theta_{k}+b_{0}*\pi,\theta_{k-1},\cdots)$, i.e., the effect of $U^k_jX_{k}$ is the same as adding a $\mathrm{X}$ gate just before rotation gate $\mathrm{RX}(\theta_{k+j*n})$.

\textbf{Example:} We consider the case where $r = 0$. In this case, $k = 0$. Since $U^0_1X_0U^0_0 = \mathrm{RX}(\theta_{n})\,\mathrm{X}\,\mathrm{RX}(\theta_{0}) = \mathrm{RX}(\theta_{n})\,\mathrm{RX}(\theta_{0})\,\mathrm{X} = U^0_1U^0_0X_0$; $U^0_1X_0U^0_0 = \mathrm{RX}(\theta_{n})\,\mathrm{X}\,\mathrm{RX}(\theta_{0})\,\mathrm{X} = \mathrm{RX}(\theta_{n})\,\mathrm{RX}(\theta_{0}) = U^0_1U^0_0$. Thus, this lemma holds when $r = 0$.

$\forall r$, $l = 2^{r+1}$, by Lemma~\ref{lemma:xpass}, we have $U^k(\bm{\theta}_{k+n})X_{k}U^k(\bm{\theta}_{k}) = U^k(\theta_{n+k}+\pi,\theta_{n+k-1},\cdots,\theta_{n})U^k(\bm{\theta}_{k}) = U^k(\theta_{n+k}+\pi,\theta_{n+k-1},\cdots,\theta_{n})U^k(\bm{\theta}_{k}) = U^k(\theta_{n+k},\theta_{n+k-1}+\pi,\cdots,\theta_{n})U^k(\theta_{k}+\pi,\theta_{k-1},\cdots,\theta_{0})$. Then, by applying Lemma~\ref{lemma:xpass} on every two layers, we can adjust gates on the $k$-th qubit of $G^k(\bm{b})$ exact the same as those of $G^{k}(\bar{\bm{b}})$.

However, we need to handle the extra $\mathrm{X}$ gate on the $k-1$-th qubit after applying Lemma~\ref{lemma:xpass}. %

We extend the notation $G^k(\bm{b})$ to $G^k(\bm{b}^k,\bm{b}^{k-1},\cdots) = U^{k}_{l-1}(\theta_{k+n*(l-1)}+b^k_{l-1}*\pi,\theta_{k-1+n*(l-1)}+b^{k-1}_{l-1}*\pi,\cdots)U^{k}_{l-2}(\theta_{k+n*(l-2)}+b^k_{l-2}*\pi,\theta_{k-1+n*(l-2)}+b^{k-1}_{l-2}*\pi,\cdots)\cdots$. And we will omit the term $\bm{b}^{j}$ if $\bm{b}^{j} = 0$. Thus, $G^k(\bm{b}^k) =G^k(\bm{b}^k,\bm{0},\cdots,\bm{0}) $. And, after applying Lemma~\ref{lemma:xpass} on every two layers, we have $G^k(\bm{b}^k) = G^{k}(\bar{\bm{b}^k}, \bm{b}^{k-1})$, where $\bm{b}^{k-1} = (b^{k-1}_{l-1},\cdots,b^{k-1}_{0}) = (1,0,1,0,\cdots,1,0)$.

We refer to the operation of cancelling every two $\mathrm{X}$ gates sequentially on the $i$-th qubit by Lemma~\ref{lemma:xcomp} and Proposition~\ref{prop:x_rule} as $P^i$. After apply $P^{k-1}$ on $G^{k}(\bar{\bm{b}^k}, \bm{b}^{k-1})$, we have $\bm{b}^{k-1} = \bm{0}$, and $\bm{b}^{k-2} = (1,1,0,0,1,1,0,0,\cdots) = p(\bm{b}^{k-1})$, here $p$ is the function induced by $P^i$ and obviously, $p$ does not depend on $i$.

We assert if $k = 2^r-1$, $l = 2^{r+1}$,  $p^{k-1}((1,0,1,0,\cdots,1,0)) = (1_{l-1},0,0,\cdots,0,1_{\frac{l}{2}-1},0,0,\cdots,0)$, $p^{2k} = (1,0,\cdots,0)$. The proof is left in the Lemma~\ref{lemma:x_string} below.

Then, if $k = 2^r - 1$, we have $G^{k}(\bm{b}^k) = G^{k}(\bar{\bm{b}^k}, \bm{b}^{k-1}) = G^{k}(\bar{\bm{b}^k}, \bm{b}^{0})$, where $\bm{b}^0 = p^{k-1}((1,0,1,0,\cdots,1,0)) = (1_{l-1},0,0,\cdots,0,1_{\frac{l}{2}-1},0,0,\cdots,0)$. Then, by applying Proposition~\ref{prop:x_rule}, $\mathrm{X}$ gate can move on the 0-th qubit, thus we have $G^{k}(\bar{\bm{b}^k}, \bm{b}^{0}) = G^{k}(\bar{\bm{b}^k})$, i.e., $G^{k}(\bm{b}^k) = G^{k}(\bar{\bm{b}^k})$.

If $2^r-1 < k \le 2^{r+1}-2$, likewise, we have $G^{k}(\bm{b}^k) = G^{k}(\bar{\bm{b}^k}, \bm{b}^{k-1}) = G^{k}(\bar{\bm{b}^k}, \bm{b}^{0})$, where $\bm{b}^0 = p^{k-1}((1,0,1,0,\cdots,1,0))$. Once $i < 2^{r+1}-2$, the $p^{i}((1,0,1,0,\cdots,1,0))$ will have an even number of 1. Thus, $\bm{b}^0$ must have even number of 1 in this case. Then, with Proposition~\ref{prop:x_rule}, we have $G^{k}(\bm{b}^k) = G^{k}(\bar{\bm{b}^k}, \bm{b}^{k-1}) = G^{k}(\bar{\bm{b}^k}, \bm{b}^{0}) = G^{k}(\bar{\bm{b}^k})$ and Lemma~\ref{lemma:xcomp} is proved.

\end{proof}

\begin{corollary}
\label{lemma:xcomp_mul}
$
    \forall k \ge 0, \forall t\in \mathbb{Z}_+, l = 2^{\lceil \log_2 (k+2) \rceil}, \forall \bm{b} \in \{0, 1\}^{t*l}, G^{k}(\bm{b}) = G^{k}(\bar{\bm{b}}).
$
\end{corollary}
\begin{proof}
Just divide $\bm{b}$ into $\bm{b}_1,\cdots \bm{b}_t$, then $G^{k}(\bm{b}) = G^{k}(\bm{b}_t)G^{k}(\bm{b}_{t-1})\cdots G^{k}(\bm{b}_1) = G^{k}(\bar{\bm{b}}_t)G^{k}(\bar{\bm{b}}_{t-1})\cdots G^{k}(\bar{\bm{b}}_1) = G^{k}(\bar{\bm{b}})$.
\end{proof}
\begin{corollary}(Translation invariance)\label{lemma:trans_invar}
For $G^{k}_i(\bm{b}) = (U^{k}_{l-1+i}X_{k}^{b_{l-1}})(U^{k}_{l-2+i}X_{k}^{b_{l-2}})\cdots (U^{k}_{i}X_{k}^{b_{0}})$, we still have
\begin{equation*}
    \forall k \ge 0, l = 2^{\lceil \log_2 (k+2) \rceil}, \forall \bm{b} \in \{0, 1\}^{l}, G^{k}_i(\bm{b}) = G^{k}_i(\bar{\bm{b}}), 
\end{equation*}
\end{corollary}
\begin{proof}
$G^{k}_i(\bm{b})$ is actually the translation of $G^{k}(\bm{b})$ with $i$ layers shifted. Then this lemma holds because of the translation invariance of RX-CX ansatz.
\end{proof}

\begin{lemma}\label{lemma:x_string}
Let $p$ denote the function induced by the operation of cancelling every two $\mathrm{X}$ gates sequentially on one qubit by Lemma~\ref{lemma:xcomp} and Proposition~\ref{prop:x_rule}. For $k = 2^r-1$ and $l = 2^{r+1}$,  $p^{k-1}((1,0,1,0,\cdots,1,0)) = (1_{l-1},0,0,\cdots,0,1_{\frac{l}{2}-1},0,0,\cdots,0)$, $p^{2k}((1,0,1,0,\cdots,1,0)) = (1,0,\cdots,0)$.
\end{lemma}
\begin{proof}
We prove it by induction.

For $r=1$, i.e., $k=1$. Look at the following sequences: \\
$p^0 = \mathrm{id}$ \\
$p((1,0,1,0)) = (1,1,0,0)$ \\
$p((1,1,0,0)) = (1,0,0,0)$ \\
Thus, this lemma holds when $r = 1$. Another example: \\
$r=2$, i.e., $k = 3$. Look at the following sequences: \\
$p((1,0,1,0,1,0,1,0)) = (1,1,0,0,1,1,0,0)$ \\
$p((1,1,0,0,1,1,0,0)) = (1,0,0,0,1,0,0,0)$ \\
$p((1,0,0,0,1,0,0,0)) = (1,1,1,1,0,0,0,0)$ \\
$p((1,1,1,1,0,0,0,0)) = (1,0,1,0,0,0,0,0)$ \\
$p((1,0,1,0,0,0,0,0)) = (1,1,0,0,0,0,0,0)$ \\
$p((1,1,0,0,0,0,0,0)) = (1,0,0,0,0,0,0,0)$ \\
Thus, this lemma holds when $r = 2$.

Suppose this lemma holds for $r=t$. For $r=t+1$, we can divide $\bm{b}$ into two halves $\bm{b}_1$ and $\bm{b}_2$. Then with the locality of $p$, we have $p^{k-1}(\bm{b}) = p^{2^{t+1}-2}(\bm{b}) = (p^{2^{t+1}-2}(\bm{b}_1),p^{2^{t+1}-2}(\bm{b}_2))$. By induction hypothesis, we have $p^{2^{t+1}-2}(\bm{b}_1) = (1,0,\cdots,0)$, thus $p^{k-1}(\bm{b}) = (1_{l-1},0,0,\cdots,0,1_{\frac{l}{2}-1},0,0,\cdots,0)$.

Further, $p^{k}(\bm{b}) = (1_{l-1},\cdots,1_{\frac{l}{2}},0,\cdots,0)$, $p^{k+1}(\bm{b}) = (1,0,\cdots,0,1_{\frac{3l}{4}-1},0,\cdots,0_{\frac{l}{2}},0,\cdots,0)$. Then, we again divide $p^{k+1}(\bm{b})$ into two halves $\bm{b}_1$ and $\bm{b}_2$. Because $p(\bm{0}) = \bm{0}$, thus, $p^{2k}(\bm{b}) = p^{k-1}(p^{k+1}(\bm{b})) = (p^{2^{t+1}-2}(\bm{b}_1), \bm{0}) = ((1,0,\cdots,0),(0,\cdots,0))=(1,0,\cdots, 0)$.

Thus, this lemma holds for $r=t+1$.

Therefore, this lemma holds for any $r \ge 1$.

\end{proof}
\begin{corollary}\label{lemma:zeronomatter}
For $\bm{a} \in \{0,1\}^l$, $p(\bm{0},\bm{a},\bm{0}) = (\bm{0},p(\bm{a}),\bm{0})$.
\end{corollary}
\begin{proof}
$p(\bm{0},\bm{a},\bm{0}) = (p(\bm{0}),p(\bm{a}),p(\bm{0})) = (\bm{0},p(\bm{a}),\bm{0})$.

Thus the heading `0's and tailing `0's can be ignored.
\end{proof}

\begin{lemma}\label{lemma:cnot_ix}
For $n$-qubit RX-CX ansatz with linear entanglement, $\mathrm{CX}(i+1, i)$ will not affect the parameter combination result on the $i$-th qubit. 
\end{lemma}
\begin{proof}
Consider the $i$-th and the $(i+1)$-th qubits in Figure~\ref{fig:rxcx}. Suppose the state of the $(i+1)$-th qubit is  
$\rho_{i+1} = p_0\qbit{0}{i+1} + p_1 \qbit{1}{i+1}$, then, $\mathrm{CX}(i+1, i)$ on target qubit $i$ is equivalent to single qubit gate $I$ with probability $p_0$ and gate $X$ with probability $p_1$.
Since both $I$ and $X$ gate can be combined into RX gate as a parameter shifting,
all $\mathrm{CX}(i+1, i)$ will not affect the result of parameter combination on the $i$-th qubit. 
Therefore, we only need to consider the effect from the $0$-th qubit to the $(i-1)$-th qubit.
\end{proof}
\begin{figure*}
    \centering
    \includegraphics[width=\textwidth]{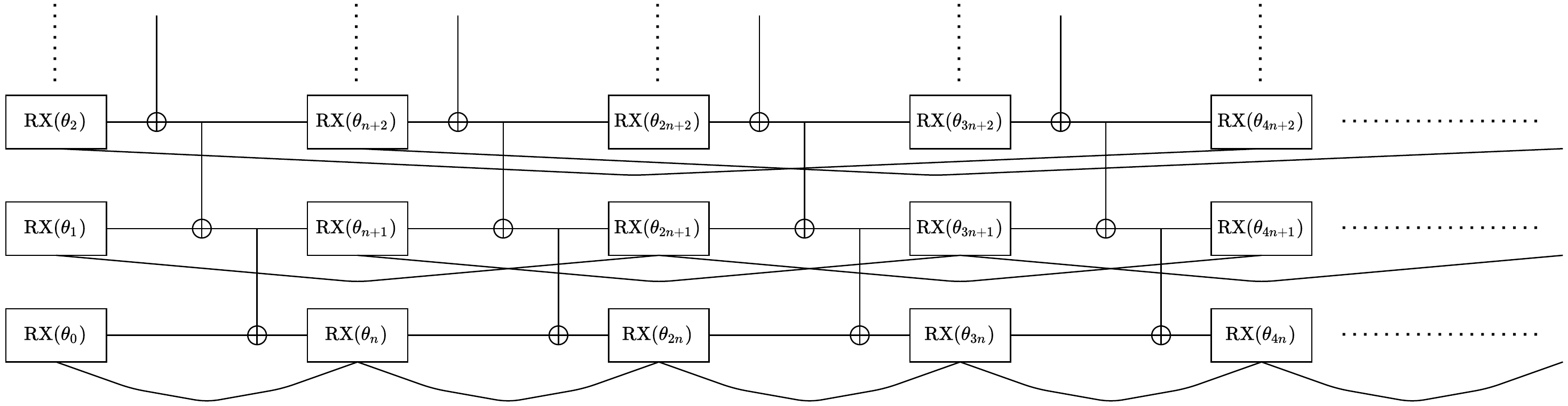}
    \caption{Illustration for parameter combination in an RX-CX ansatz with linear entanglement} %
    \label{fig:rxcx}
\end{figure*}

\begin{lemma}
An $n$-qubit RX-CX ansatz with linear entanglement, the parameters $\theta_{1+n*j}$ (the parameter in the $j$-th RX gate on the 1st qubit) and $\theta_{1+n*j+n*2^{\lceil \log_2 2 \rceil}}$ can be combined.

\end{lemma}
\begin{proof}
Without loss of generality, we prove parameters $\theta_{1}$ and $\theta_{1+n*2^{\lceil \log_2 2 \rceil}}$ can be combined.

From Equation~(\ref{equ:recur_layer}) (in Lemma~\ref{lemma:layer_uni}), we have
$$U_1 = \begin{pmatrix}
C_{1}U^{0}(\bm{\theta}_{0}) & -iS_{1}U^{0}(\bm{\theta}_{0}) \\
    -iS_{1}U^{0}(\bm{\theta}_{0})X_{0} & C_{k}U^0(\bm{\theta}_{0})X_{0}  
\end{pmatrix}$$. \\
Then, we have:
$$U^1_{1+n}U^1_1 = \begin{pmatrix}
C_{1+n}C_1 U^0_n U^0_0 - S_{1+n}S_1 U^0_n U^0_0X_0 & -iC_{1+n}S_1 U^0_n U^0_0 - iS_{1+n}C_1 U^0_n U^0_0X_0 \\
-iS_{1+n}C_1U^0_nX_0U_0^0 - iC_{1+n}S_1U^0_nX_0U^0_0X_0 & -S_{1+n}S_1U_n^0X_0U_0^0 + C_{1+n}C_1U_n^{0}X_0U_0^0X_0
\end{pmatrix} $$

Without loss of generality, we show $\theta_{1}$ and $\theta_{1+2n}$ can be combined by consider the first entry of $U^1_{1+n*2^1}U^1_{1+n}U^1_1$:
$$C_{1+2n}C_{1+n}C_1U^0_{2n}U^0_{n}U^0_0 - S_{1+2n}C_{1+n}S_{1}U^0_{2n}U^0_{n}X_0U^0_0X_0 - C_{1+2n}S_{1+n}S_1U^0_{2n}U^0_{n}U^0_0X_0 - S_{1+2n}S_{1+n}C_1U^0_{2n}U^0_{n}X_0U_0^0$$

For the first two terms, we have:
$$C_{1+2n}C_{1+n}C_1U^0_{2n}U^0_{n}U^0_0 - S_{1+2n}C_{1+n}S_{1}U^0_{2n}U^0_{n}X_0U^0_0X_0 = C_{1+n}(C_{1+2n}C_1U^0_{2n}U^0_{n}U^0_0 - S_{1+2n}S_{1}U^0_{2n}U^0_{n}X_0U^0_0X_0)$$

To combine $\theta_{1+2n}$ and $\theta_1$ into $\theta_{1+2n}+\theta_1$, we only need to show $U^0_nX_0U^0_nX_0 = U^0_nU^0_0$. This can be shown by applying Lemma~\ref{lemma:xcomp} immediately.

We also have $ U_n^0 U_0^0 X_0 = U_n^0X_0 U_0^0 $ with Lemma~\ref{lemma:xcomp}. Thus the remaining two terms can be also combined, and parameter $\theta_{1+2n}$ and $\theta_1$ will be combined into $\theta_{1+2n}+\theta_1$. \\
Thus, $\theta_{1+2n}$ and $\theta_1$ can be combined into $\theta_{1+2n}+\theta_1$ for the first entry of $U^1_{1+n*2^1}U^1_{1+n}U^1_1$. The parameter combination for other entries can be proved similarly.

\end{proof}

\begin{lemma}\label{lemma:layer_cnt}
For $n$-qubit RX-CX ansatz with linear entanglement, the two parameters $\theta_{i+n*j}$ (the parameter in the $j$-th $\rm RX$ gate on the $i$-th qubit) and $\theta_{i + n*j + n*2^{\lceil \log_2 (i+1)\rceil}}$ on the $i$-th qubit can be combined.
\end{lemma}
\begin{proof}

With Lemma~\ref{lemma:cnot_ix}, we only need to consider the effect from the 0-th qubit to the $(i-1)$-th qubit when computing the parameter combination results on $i$-th qubit.

We prove it by induction. For qubit $i = 0$ of $n$-qubit ansatz, we already have that $\theta_{i}$ and $\theta_{i+n*2^{\lceil \log_2 (i+1)\rceil}}$ can be combined by Proposition~\ref{prop:x_rule}.

Assume that when $i = k - 1$, the two parameters $\theta_{i+n*j}$ (the parameter in the $j$-th $\rm RX$ gate on the $i$-th qubit) and $\theta_{i + n*j + n*2^{\lceil \log_2 (i+1)\rceil}}$ on the $i$-th qubit can be combined.

So, for $i = k$, 
without loss of generality, we only prove that parameters $\theta_{k}$ and $\theta_{k+n*2^{\lceil \log_2 (k+1)\rceil}}$ can be combined. 

From Lemma~\ref{lemma:layer_uni}, the unitary of the first layer is 
\begin{equation}
 U^k(\bm{\theta}_k) = \begin{pmatrix}
C_{k}U^{k-1}(\bm{\theta}_{k-1}) & -iS_{k}U^{k-1}(\bm{\theta}_{k-1}) \\
    -iS_{k}U^{k-1}(\bm{\theta}_{k-1})X_{k-1} & C_{k}U^{k-1}(\bm{\theta}_{k-1})X_{k-1}  
\end{pmatrix}
\end{equation}

Then, consider the composition of the first $l = 2^{\lceil \log_2 (k+1) \rceil}$ layers: $\prod\limits_{i = 0}^{l} U^k(\bm{\theta}_{n*i+k})$. For simplicity, in the following, we denote $U^k(\bm{\theta}_{n*i+k})$ by $U^k_{n*i+k}$. \\
Without loss of generality, we still consider the (0, 0) entry of $\prod\limits_{i = 0}^{l} U^k(\bm{\theta}_{n*i+k})$. 
The (0, 0) entry is the sum of many terms during the matrix multiplication. By induction (details are left in Lemma~\ref{lemma:layer_comb} below), we know every term of the (0, 0) entry of $\prod\limits_{i = 0}^{l} U^k(\bm{\theta}_{n*i+k})$ obeys the following constraints: \\
(1) every term is of the form: $\prod\limits_{i=0}^{l} T_{a_i}(\theta_{n*i+k})G_{b_i}^{k-1}(\bm{\theta}_{n*i+k-1})$, where $a_i \in \{0, 1\}$, $T_{a_i}((\theta_{n*i+k})) = (C_{n*i+k})^{1-a_i}(-iS_{n*i+k})^{a_i}$; $b_i \in \{ 0, 1 \}$,  $G_{b_i}^{k-1}(\bm{\theta}_{n*i+k-1}) = (U^{k-1}(\bm{\theta}_{n*i+k-1})X_{k-1})^{b_i} (U^{k-1}(\bm{\theta}_{n*i+k-1}))^{1-b_i}$. \\
(2) For every term in the such entry,  $\bm{b}=(b_0, b_1, \cdots, b_{l})$ is decided by $\bm{a}=(a_0, a_1, \cdots, a_{l})$ by the following process:
$b_0 = a_0$; $b_i = \text{XOR}(b_{i-1}, a_i)$ if $i \ge 1$. Thus, we can represent a term by only using vector $\bm{a}$.

For simplicity, we define some operations on vector $\bm{a}$ or $\bm{b}$: \\
(1) \textbf{Complement}: $\bar{\bm{a}} = (1-a_0, \cdots, 1-a_{l})$.
 \\
(2) \textbf{Length}: $l(\bm{a})$ (or $l(\bm{b})$) refers the number of components in vector $\bm{a}$ (resp. $\bm{b}$). \\
(3) \textbf{Slicing}: $\bm{a}(n:m) = (a_n, a_{n+1}, \cdots, a_m)$. \\
(4) \textbf{Indexing}: $\bm{a}(i) = a_i$. \\
(5) \textbf{Index set}: $s(\bm{a})$ is the lowest index of vector $\bm{a}$, $u(\bm{a})$ is the highest index of $\bm{a}$.

We denote $\prod\limits_{i=s(\bm{a})}^{u(\bm{a})}$
$T_{a_i}(\theta_{n*i+k})G_{b_i}^{k-1}(\bm{\theta}_{n*i+k-1})$ by $R(\bm{a})$; denote $\prod\limits_{i=s(\bm{b})}^{u(\bm{b})} G_{b_i}^{k-1}(\bm{\theta}_{n*i+k-1})$ by $G(\bm{b})$.

Then, we show how terms in the (0, 0) entry combines. Consider a pair of terms $R(\bm{a}_1)$ and $R(\bm{a}_2)$ where $\bm{a}_1(0) = 1 - \bm{a}_2(0)$, $\bm{a}_1(l) = 1 - \bm{a}_2(l)$, and $\bm{a}_1(i) = \bm{a}_2(i)$, $\forall i \in (0, l)$. We prove that if $R(\bm{a}_1)$ is in the (0, 0) entry, then $R(\bm{a}_2)$ must be also in the (0, 0) entry. \\
Actually, if $\bm{a}_1(l) = 0$, since $\prod\limits_{i = 0}^{l} U^k(\bm{\theta}_{n*i+k}) = U^k(\bm{\theta}_{n*l+k})\prod\limits_{i = 0}^{l-1} U^k(\bm{\theta}_{n*i+k})$, $R(\bm{a}_1(0:l-1))$ must be in the (0,0) entry of $\prod\limits_{i = 0}^{l-1} U^k(\bm{\theta}_{n*i+k})$. From Lemma~\ref{lemma:layer_join} below, we know that $R(\bm{a}_2(0:l-1))$ is in the (1,0) entry of $\prod\limits_{i = 0}^{l-1} U^k(\bm{\theta}_{n*i+k})$. Thus, $R(\bm{a}_2) = -iS_{n*l+k}U^{k-1}_{n*l+k-1}R(\bm{a}_2(0:l-1))$ is in the (0, 0) entry of $\prod\limits_{i = 0}^{l} U^k(\bm{\theta}_{n*i+k})$. We can prove the case $\bm{a}_1(l) = 1$ similarly. Thus, if $R(\bm{a}_1)$ is in the (0, 0) entry, then $R(\bm{a}_2)$ must be also in the (0, 0) entry. \\
Then, $R(\bm{a}_1) + R(\bm{a}_2) = T_{\bm{a}_1(0)}T_{\bm{a}_1(l)}\prod\limits_{i=1}^{l-1} T_{\bm{a}_1(i)}(\theta_{n*i+k})G(\bm{b}_1) + T_{\bm{a}_2(0)}T_{\bm{a}_2(l)}\prod\limits_{i=1}^{l-1} T_{\bm{a}_2(i)}(\theta_{n*i+k})G(\bm{b}_2)$. \\
Note that $T_{\bm{a}_1(0)}(\theta_{k})T_{\bm{a}_1(l)}(\theta_{n*l+k}) + T_{\bm{a}_2(0)}(\theta_{k})T_{\bm{a}_2(l)}(\theta_{n*l+k}) = T_s(\theta_{k}+\theta_{n*l+k})$, where $s = \min (\bm{a}_1(0)+\bm{a}_1(l), \bm{a}_2(0)+\bm{a}_2(l))$, then once $G(\bm{b}_1) = G(\bm{b}_2)$, for $R(\bm{a}_1) + R(\bm{a}_2)$, $\theta_{k}$ and $\theta_{n*l+k}$ can be combined into $\theta_{k}+\theta_{n*l+k}$.

Actually, since $\bm{a}_1(0)$ and $\bm{a}_2(0)$ complements, we have $\bm{b}_1(0:l-1) = \overline{\bm{b}_2(0:l-1)}$.
Note that $\bm{b}_1(l) = \bm{b}_2(l)$, we only need to consider $G(\bm{b}_1(0,l-1)) = G(\bm{b}_2(0,l-1))$. 

Note that $l=2^{\lceil \log_2 (k+1) \rceil} = 2^{\lceil \log_2 ((k-1)+2) \rceil}$, then with Lemma~\ref{lemma:xcomp} and Corollary~\ref{lemma:xcomp_mul}, we have $G(\bm{b}_1(0:l-1)) = G(\bm{b}_2(0:l-1))$. This means parameters $\theta_{n*l+k}$ and $\theta_{k}$ can be combined in the (0, 0) entry. We can prove this result on other entries similarly. \\
Therefore, the two parameters $\theta_{i+n*j}$ (the parameter in the $j$-th $\rm RX$ gate on the $i$-th qubit) and $\theta_{i + n*j + n*2^{\lceil \log_2 (i+1)\rceil}}$ on the $i$-th qubit can be combined.
\end{proof}

\begin{lemma}\label{lemma:layer_comb}
$\forall r \ge 1$, every term of the (0, 0) or (1, 0) entry of $\prod\limits_{i = 0}^{r} U^k(\bm{\theta}_{n*i+k})$ obeys the following constraints: \\
(1) Every term is of the form: $\prod\limits_{i=0}^{r} T_{a_i}(\theta_{n*i+k})G_{b_i}^{k-1}(\bm{\theta}_{n*i+k-1})$, where $a_i \in \{0, 1\}$, $T_{a_i}((\theta_{n*i+k})) = (C_{n*i+k})^{1-a_i}(-iS_{n*i+k})^{a_i}$; $b_i \in \{ 0, 1 \}$,  $G_{b_i}^{k-1}(\bm{\theta}_{n*i+k-1}) = (U^{k-1}(\bm{\theta}_{n*i+k-1})X_{k-1})^{b_i} (U^{k-1}(\bm{\theta}_{n*i+k-1}))^{1-b_i}$. \\
(2) For every term in such entry, $\bm{b}= (b_0, b_1, \cdots, b_{r})$ is decided by $\bm{a}= (a_0, a_1, \cdots, a_{r})$ by the following process:
$b_0 = a_0$; $b_i = \mathrm{XOR}(b_{i-1}, a_i)$ if $i \ge 1$. For simplicity, we use $\bm{p}$ to represent this relation, i.e., $\bm{b} = \bm{p}(\bm{a})$. Thus, we can represent a term by using only vector $\bm{a}$.
\end{lemma}
\begin{proof}
The first conclusion of the lemma is obvious and we focus on the second one. We can prove it by induction.

We first consider $r=1$. The (0, 0) entry is $C_{k+n}C_k U^{k-1}_{n+k-1} U^{k-1}_{k-1} - S_{1+n}S_1 U^{k-1}_{n+k-1} U^{k-1}_{k-1}X_{k-1}$, the (1,0) entry is $-iS_{k+n}C_k U^{k-1}_{n+k-1} X_{k-1} U^{k-1}_{k-1} -iC_{1+n}S_1 U^{k-1}_{n+k-1}X_{k-1} U^{k-1}_{k-1}X_{k-1}$.  It is easy to verify that this lemma holds when $r = 1$.

Then we assume this lemma holds when $r = t$.

Then for $r = t+1$, we have $\prod\limits_{i = 0}^{t+1} U^k(\bm{\theta}_{n*i+k}) = U^k(\bm{\theta}_{n*(t+1)+k})\prod\limits_{i = 0}^{t} U^k(\bm{\theta}_{n*i+k}) = \begin{pmatrix}
C_{n*(t+1)+k}U^{k-1}_{n*(t+1)+k} & -iS_{n*(t+1)+k}U^{k-1}_{n*(t+1)+k} \\
-iS_{n*(t+1)+k}U^{k-1}_{n*(t+1)+k}X_{k-1} & C_{n*(t+1)+k}U^{k-1}_{n*(t+1)+k}X_{k-1} 
\end{pmatrix} \prod\limits_{i = 0}^{t} U^k(\bm{\theta}_{n*i+k})$.

Since $\prod\limits_{i = 0}^{t} U^k(\bm{\theta}_{n*i+k}) = U^k(\bm{\theta}_{n*t+k})\prod\limits_{i = 0}^{t-1} U^k(\bm{\theta}_{n*i+k})$, we know that for every term in the (0, 0) entry of $\prod\limits_{i = 0}^{t} U^k(\bm{\theta}_{n*i+k})$, we must have $b_t = 0$ because $G^{k-1}_{b_{t}}(\bm{\theta}_{n*t+k-1}) = U^{k-1}(\bm{\theta}_{n*t+k-1})$. Likewise, we know that for every term in the (1,0) entry of $\prod\limits_{i = 0}^{t} U^k(\bm{\theta}_{n*i+k})$, we must have $b_t = 1$.

For every term in the (0, 0) entry of $\prod\limits_{i = 0}^{t} U^k(\bm{\theta}_{n*i+k})$, denote its $\bm{a}$-vector by $\bm{a}_t$ and its $\bm{b}$-vector by $\bm{b}_t$. After multiplying with $C_{n*(t+1)+k}U^{k-1}_{n*(t+1)+k}$, we have vector $\bm{a}_{t+1} = (\bm{a}_t , 0)$, $\bm{b}_{t+1} = (\bm{b}_t , 0)$. Thus, $\bm{b}_{t+1}(t+1) = 0 = \mathrm{XOR}(\bm{a}_{t+1}(t+1),\bm{b}_{t+1}(t))$. And since we already know $\bm{b}_t = \bm{p}(\bm{a}_t)$, we now have $\bm{b}_{t+1} = \bm{p}(\bm{a}_{t+1})$. Likewise, for every term in the (1,0) entry of $\prod\limits_{i = 0}^{t} U^k(\bm{\theta}_{n*i+k})$, after multiplying with $-iS_{n*(t+1)+k}U^{k-1}_{n*(t+1)+k}$, we have vector $\bm{a}_{t+1}(t+1)=1$. Since $\bm{b}_{t+1}(t) = \bm{b}_t(t) = 1$, we have $\bm{b}_{t+1}(t+1) = 0 = \mathrm{XOR}(\bm{a}_{t+1}(t+1), \bm{b}_t(t) = 1)$. Thus, in this case, we still have $\bm{b}_{t+1} = \bm{p}(\bm{a}_{t+1})$.
Combining these two parts, we have terms in (0, 0) entry satisfies the two constraints.

Likewise, we can prove the lemma for (1,0) entry. 
\end{proof}

\begin{lemma}\label{lemma:layer_join}
$\forall r \ge 1$, if $R(\bm{a}_1) = \prod\limits_{i=0}^{r} T_{\bm{a}_1(i)}(\theta_{n*i+k})G(\bm{b}_1)$ is in the (0, 0) entry of $\prod\limits_{i = 0}^{r} U^k(\bm{\theta}_{n*i+k})$, then $R(\bm{a}_2)$ must be in the (1, 0) entry, where $\bm{a}_2(0) = 1-\bm{a}_2(0)$, $\bm{a}_2(i)=\bm{a}_1(i)$, $\forall i \in (0, r]$.
\end{lemma}
\begin{proof}
We prove it by induction. \\
We first consider $r = 1$. The (0, 0) entry is $C_{k+n}C_k U^{k-1}_{n+k-1} U^{k-1}_{k-1} + (-iS_{1+n})(-iS_1) U^{k-1}_{n+k-1} U^{k-1}_{k-1}X_{k-1}$, the (1,0) entry is $C_{1+n}(-iS_1 )U^{k-1}_{n+k-1}X_{k-1} +(-iS_{k+n})C_k U^{k-1}_{n+k-1} X_{k-1} U^{k-1}_{k-1}  U^{k-1}_{k-1}X_{k-1}$.  It is easy to verify that this lemma holds when $r = 1$. 

Then we assume this lemma holds when $r=t$. \\
Then, for $r = t + 1$, without loss of generality, we assume $\bm{a}_1(t+1) = 0$. If $R(\bm{a}_1) = \prod\limits_{i=0}^{t+1} T_{\bm{a}_1(i)}(\theta_{n*i+k})G(\bm{b}_1)$ is in the (0, 0) entry of $\prod\limits_{i = 0}^{t+1} U^k(\bm{\theta}_{n*i+k})$, then $R(\bm{a}_1(0:t))$ must be in the (0, 0) entry of $\prod\limits_{i = 0}^{t} U^k(\bm{\theta}_{n*i+k})$. Then, by inductive hypothesis, we have $R(\bm{a}_2(0:t))$ must be in the (1, 0) entry of $\prod\limits_{i = 0}^{t} U^k(\bm{\theta}_{n*i+k})$. Since $\prod\limits_{i = 0}^{t+1} U^k(\bm{\theta}_{n*i+k}) = U^k(\bm{\theta}_{n*(t+1)+k})\prod\limits_{i = 0}^{t} U^k(\bm{\theta}_{n*i+k})$, we have $R(\bm{a}_2) = C_{n*(t+1)+k}U^{k-1}(\bm{\theta}_{n*(t+1)+k-1})X_{k-1}R(\bm{a}_2(0:t))$ being in the (1, 0) entry. Thus this lemma holds when $r = t+1$.

To conclude, this lemma holds for any $r\ge 1$.
\end{proof}

\textbf{Proposition 3.2.}  
\textit{Suppose the number of effective parameters of an $n$-qubit RX-CX ansatz with linear entanglement is $K$. We have the following conclusion:}
\begin{align*}
    K \le \lfloor\frac{3n^2+1}{4}\rfloor
\end{align*}
\begin{proof}
From Lemma~\ref{lemma:layer_comb}, there are at most $2^{\lceil \log_2 (i+1) \rceil}$ effective parameters on $i$-th qubit. Assume $2^r < n \le 2^{r+1}$. Then the total number of effective parameters of linear $n$-qubit RX-CX ansatz is at most 
\begin{equation}
    f(n) = 1 + \sum_{i=1}^r (2^{t-1})2^t + (n - 2^r)*2^{r+1} = \frac{1}{3} - \frac{4^{r+1}}{3} + n*2^{r+1}
\end{equation}

For given $n$, for $r \in [\log_2n-1, \log_2n)$, the maximum value is achieved when $2^{r+1} = \frac{3n}{2}$, i.e., $f(n) \le \frac{3}{4}n^2 + \frac{1}{3}$. Since $f(n)$ is integer, thus we have $K = f(n) \le \lfloor\frac{3n^2+1}{4}\rfloor$.

\end{proof}

\subsection{Proposition~\ref{prop:rxcx_cnotorder} and Corollary~\ref{coro:rxcx-AL-layer}: Order of the CX Gates}\label{sec:order_cnot}

We can extend the result of Lemma~\ref{lemma:xcomp} to a more general case. Consider an entanglement layer which is same as the linear entanglement layer except that the order of CX gates may be different. We denote such entanglement layer as modified linear entanglement layer. Then for any modified linear entanglement layer, we always have $G^{k}(\bm{b}) = G^{k}(\bar{\bm{b}})$.

\begin{lemma}\label{lemma:cnotnomatter}
For $n$-qubit RX-CX ansatz with modified linear entanglement, and $G^{k}(\bm{b}) = (U^{k}_{l-1}X_{k}^{b_{l-1}})(U^{k}_{l-2}X_{k}^{b_{l-2}})\cdots (U^{k}_{0}X_{k}^{b_{0}})$ we still have
\begin{equation*}
    \forall k \ge 0, l = 2^{\lceil \log_2 (k+2) \rceil}, \forall \bm{b} \in \{0, 1\}^{l}, G^{k}(\bm{b}) = G^{k}(\bar{\bm{b}}),
\end{equation*}
where $U^k_i$ is the unitary of one layer from 0-th qubit to $k$-th qubit with gates involved with other qubits ignored.
\end{lemma}
\begin{proof}
Actually, we only need to consider the following two cases: \\
1. $\mathrm{CX}(k,k-1)$ is before $\mathrm{CX}(k-1,k-2)$.
\begin{equation*}
    \includegraphics[height=6\fontcharht\font`\B]{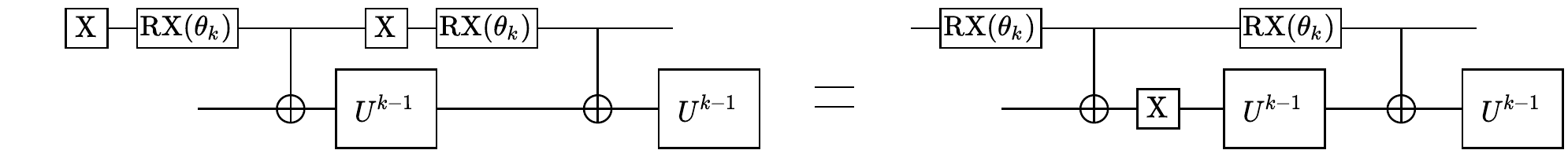} 
\end{equation*}

2. $\mathrm{CX}(k,k-1)$ is after $\mathrm{CX}(k-1,k-2)$. 
\begin{equation*}
    \includegraphics[height=6\fontcharht\font`\B]{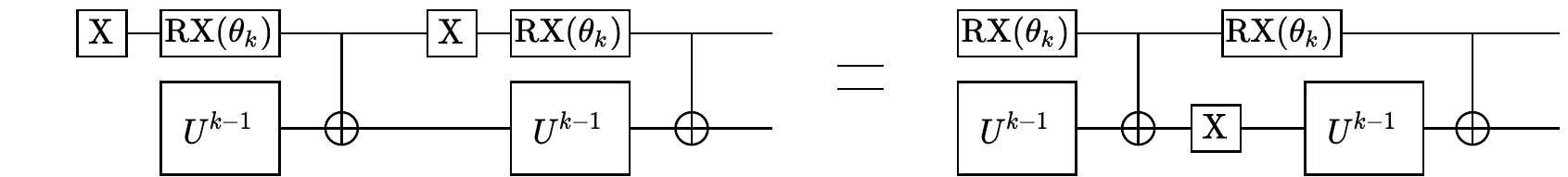} 
\end{equation*}
\end{proof}

For case 1, $p((1,1)) = (1,0)$; for case 2, $p((1,1)) = (0,1)$.

To avoid notation ambiguity, we refer the function induced by the operation of cancelling every two $\mathrm{X}$ gates sequentially on case 1 (resp. case 2) by Lemma~\ref{lemma:xcomp} and Proposition~\ref{prop:x_rule} as $p_1$ (resp. $p_2$).

Likewise, to show $G^{k}(\bm{b}) = G^{k}(\bar{\bm{b}})$, we only need to show: \\
\textbf{Proposition}: for any vector $\bm{x} \in \{ 0,1 \}^{l_1}$,  $p_{\bm{x}} = p_{x_{0}}\circ p_{x_1} \circ \cdots \circ p_{x_{l_1-1}}$,  $p_{\bm{x}}((1,0,1,0,\cdots,1,0))$ contains an odd number of $1$ if and only if $l_1 = l-2$ ($l$ is the length of input), and in that case, $p_{\bm{x}}((1,0,1,0,\cdots,1,0))$ contains only one $1$.

Note that $(0,1,\cdots,0,1)$ can be transformed into $(1,0,\cdots,1,0)$ by Corollary~\ref{lemma:zeronomatter} without changing the conclusion of this proposition. Thus, in the proposition, we only consider initial vector $(1,0,\cdots,1,0)$.

We can prove the proposition by induction on $l = 2^{r+1}$.

We first consider $r = 1$. if $l_1 < l-2 = 2$, then $l_1$ must be $0,1,2,3$. \\
If $l_1 = 0$, $p_{\bm{x}} = \mathrm{id}$, the proposition already holds.
If $l_1 = 1$, $p_{\bm{x}}$ only produces: $(1,1,0,0)$ and $(0,1,1,0)$, still holds. \\
Thus this proposition holds when $r = 1$.

Then we assume this proposition holds when $r = t > 1$. Then, for $r = t + 1$, $l=2^{t+2}$, if $l_1 \le 2^{t+1} - 2$, we can cut $(1,0,1,0,\cdots)_{l}$ into two halves. Then  $p_{\bm{x}}((1,0,1,0,\cdots)_{l}) = (p_{\bm{x}}((1,0,1,0,\cdots)_{l/2}), p_{\bm{x}}((1,0,1,0,\cdots)_{l}/2))$ must have even number of `1'.

Because for $l_1 = 2^{t+1} - 2$, $p_{\bm{x}}((1,0,1,0,\cdots)_{l/2})$ has only one 1, thus $p_{\bm{x}}((1,0,1,0,\cdots)_{l})$ have two `1' with a distance $\frac{l}{2}$. \\
Then, for $l_1 = 2^{t+1} - 1$, no matter $x_{l_1-1}$ is 0 or 1, $p_{\bm{x}}((1,0,1,0,\cdots)_{l})$ must contain $\frac{l}{2}$ consecutive `1' and remaining elements is `0'. In this case, we still have even number of `0'.

Since for both $p_0$ and $p_1$, heading `0's and tailing `0's can be ignored (Corollary~\ref{lemma:zeronomatter}), thus, we can get vector $(1,0,1,0,\cdots,1,0,1)_{\frac{l}{2}-1}$ after applying $p_{\bm{x}}$ when $l_1 = 2^{t+1}$. In this case, we still have an even number of `0'.

By Corollary~\ref{lemma:zeronomatter}, we can pad $(1,0,1,0,\cdots,1,0,1)_{\frac{l}{2}-1}$ to $(1,0,1,0,\cdots,1,0,1,0)_{\frac{l}{2}}$ without affecting the proposition here. 

Then, by induction, we need $l_2 = 2^{t+1} -2$, $\bm{x}_2\in \{0, 1\}^{l_2}$, $p_{\bm{x}_2}$ to make $(1,0,1,0,\cdots,1,0,1,0)_{\frac{l}{2}}$ become a vector with only one 1.

Thus, we need $l_1 = 2^{t+1} -2 + 2^{t+1} = 2^{t+2} -2 = l -2$ to let $p_{\bm{x}}((1,0,1,0,\cdots,1,0))$ contain an odd number of `1', i.e., this proposition holds when $r = t + 1$.

Thus, this proposition holds for any $r \ge 1$.

RX-CX ansatz with modified linear entanglement has the same parameter combination result as the RX-CX ansatz with original linear entanglement, i.e., 

\textbf{Proposition 3.3} (CX order does not matter)\textbf{.} 
\textit{An $n$-qubit RX-CX ansatz with modified linear entanglement has at most $\lfloor\frac{3n^2+1}{4}\rfloor$ effective parameters. The modified linear entanglement has alternated CX order compared with the original linear entanglement as shown in Figure~\ref{fig:entanglementblock} (b)(c).}
\begin{proof}
There are at most two types of form on $k$-th qubit: $U^{k}(\bm{\theta}_{k}) = U^{k-1}(\bm{\theta}_{k-1}) \cdot \mathrm{CX}(k,k-1) \cdot \mathrm{RX}(\theta_k)\otimes \mathrm{I}$ or $U^{k}(\bm{\theta}_{k}) = \mathrm{CX}(k,k-1) \cdot U^{k-1}(\bm{\theta}_{k-1}) \cdot  \mathrm{RX}(\theta_k)\otimes \mathrm{I}$. Since we have already proved the parameter combination of the first form, here we will focus on the second one, i.e.
\begin{align}\label{equ:recur_layer1}
 U^k(\bm{\theta}_k) = \begin{pmatrix}
C_{k}U^{k-1}(\bm{\theta}_{k-1}) & -iS_{k}U^{k-1}(\bm{\theta}_{k-1}) \\
    -iS_{k}X_{k-1}U^{k-1}(\bm{\theta}_{k-1}) & C_{k}X_{k-1}U^{k-1}(\bm{\theta}_{k-1})
\end{pmatrix}
\end{align}

Likewise, we first prove the two parameters $\theta_{k+n*j}$ (the parameter in the $j$-th $\rm RX$ gate on the $i$-th qubit) and $\theta_{k + n*j + n*2^{\lceil \log_2 (k+1)\rceil}}$ on the $k$-th qubit can be combined.

We follow the notations in the proof of Lemma~\ref{lemma:layer_cnt}. 

With loss of generality, we only consider the (0, 0) entry. Let $l = 2^{\lceil \log_2 (k+1)\rceil}$. In this proposition, each term of (0, 0) entry is like $\prod\limits_{i=0}^{l} T_{a_i}(\theta_{n*i+k})W_{b_i}^{k-1}(\bm{\theta}_{n*i+k-1})$, where $a_i \in \{0, 1\}$, $T_{a_i}((\theta_{n*i+k})) = (C_{n*i+k})^{1-a_i}(-iS_{n*i+k})^{a_i}$; $b_i \in \{ 0, 1 \}$,  $W_{b_i}^{k-1}(\bm{\theta}_{n*i+k-1}) = (X_{k-1}U^{k-1}(\bm{\theta}_{n*i+k-1}))^{b_i} (U^{k-1}(\bm{\theta}_{n*i+k-1}))^{1-b_i}$. Denote $\prod\limits_{i=s(\bm{b})}^{u(\bm{b})} W_{b_i}^{k-1}(\bm{\theta}_{n*i+k-1})$ by $W(\bm{b})$. 

Similar to the proof of Lemma~\ref{lemma:layer_cnt}, to show parameter $\theta_{k}$ and $\theta_{n*l+k}$ can be combined, we only need to prove: $U^{k-1}(\bm{\theta}_{n*l+k-1})W(\bm{b}_1(0:l-1)) = U^{k-1}(\bm{\theta}_{n*l+k-1})W(\overline{\bm{b}_1(0:l-1)})$.

We define $\bm{b}_2$ as follows: \\
(1) $\bm{b}_2(0) = 0$. \\
(2) $\bm{b}_2(i+1) = \bm{b}_1(i)$, $\forall i \in (0, l)$.

Then, we have $U^{k-1}(\bm{\theta}_{n*l+k-1})W(\bm{b}_1(0:l-1)) = G(\bm{b}_2(1:l))U^{k-1}(\bm{\theta}_{k-1})$, $U^{k-1}(\bm{\theta}_{n*l+k-1})W(\overline{\bm{b}_1(0:l-1)}) = G(\overline{\bm{b}_2(1:l)})U^{k-1}(\bm{\theta}_{k-1})$.

Since $l=2^{\lceil \log_2 (k+1) \rceil} = 2^{\lceil \log_2 ((k-1)+2) \rceil}$, with Lemma~\ref{lemma:cnotnomatter},  Corollary~\ref{lemma:xcomp_mul} and Corollary~\ref{lemma:trans_invar}, we have $U^{k-1}(\bm{\theta}_{n*l+k-1})W(\bm{b}_1(0:l-1)) = G(\bm{b}_2(1:l))U^{k-1}(\bm{\theta}_{k-1}) = G(\overline{\bm{b}_2(1:l)})U^{k-1}(\bm{\theta}_{k-1}) = U^{k-1}(\bm{\theta}_{n*l+k-1})W(\overline{\bm{b}_1(0:l-1)})$.

Therefore, the two parameters $\theta_{i+n*j}$ (the parameter in the $j$-th $\rm RX$ gate on the $i$-th qubit) and $\theta_{i + n*j + n*2^{\lceil \log_2 (i+1)\rceil}}$ on the $i$-th qubit can be combined. This means that we have at most $2^{\lceil \log_2 (i+1)\rceil}$ effective parameters on $i$-th qubit.

Similar to Proposition~\ref{prop:rx_cx_upperbound}, we have that $n$-qubit RX-CX ansatz with modified linear entanglement has at most $\lfloor\frac{3n^2+1}{4}\rfloor$ effective parameters.

\end{proof}

\begin{figure*}
    \centering
    \includegraphics[width=\textwidth]{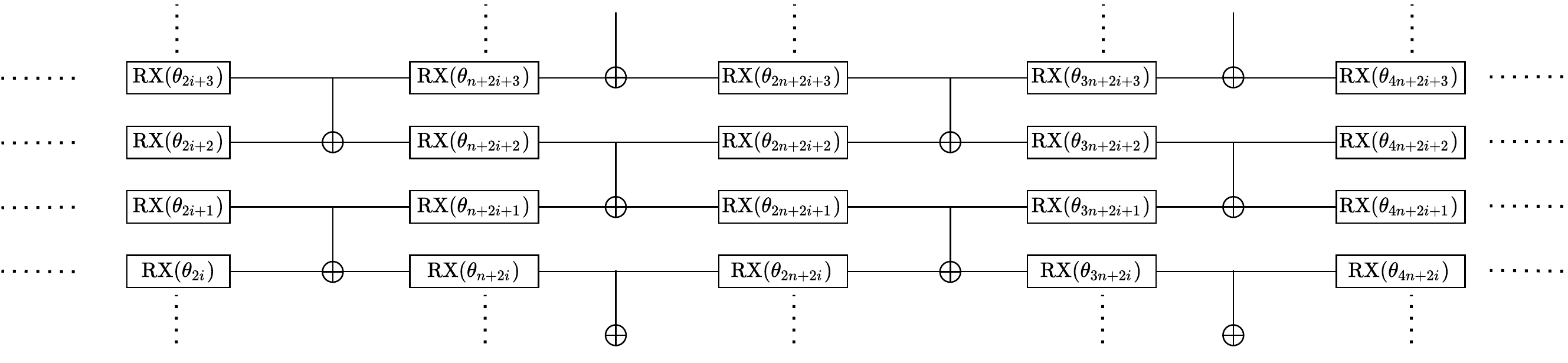}
    \caption{RX-CX ansatz with alternating entanglement}
    \label{fig:rxcx_alter}
\end{figure*}

\textbf{Corollary 3.1.}
A $(2L)$-layer RX-CX ansatz with alternating entanglement can be reduced to an $L$-layer RX-CX ansatz with linear entanglement.
\begin{proof}
X-rule can directly exchange RX and CX on target qubit.

As shown in Figure~\ref{fig:rxcx_alter}, by applying the X-rule repeatedly, we can combine $\theta_{n+i+2k}$ to $\theta_{i+2k}$, $\theta_{n+i+2k+1}$ to $\theta_{2n+i+2k+1}$, ... Then within finite steps, the RX gate between alternating entanglement layers will be all eliminated, and we get an entanglement layer that is same as linear entanglement, regardless of the order of CX.

\end{proof}

\subsection{Proposition~\ref{prop:rep_deg} and Proposition~\ref{prop:limit_cx}: More CX Gates in the Entanglement Layer}\label{sec:more_cnot}

\textbf{Proposition 3.4} (Repetition of entanglement layer)\textbf{.}
\textit{For any entanglement layer that is purely constructed by CX gates, there exists an integer $k$, s.t. $E^k = I$ where $E$ is the unitary transformation represented by the entanglement layer and $I$ is the identity operator.}
\begin{proof}
Because $E$ is just a special permutation matrix, and the number of $2^n$-dimensional permutation matrix is limited. Thus $\{E, E^2, E^3 \cdots,\}$ is a cyclic group, i.e., there would be some $k$, s.t. $E^k = I$.
\end{proof}

\textbf{Proposition 3.5} (Limitation of CX-based entanglement layer)\textbf{.} 
\textit{For a given $n$-qubit unitary transformation $U$, $\forall i, j \ge 2$, it is possible to move $i$-th column of $U$ to $j$-th column of $U$ with a finite number of CXs, but it is not possible to achieve column swap between two arbitrary columns when $n \ge 3$.}
\begin{proof}
Consider the form: $U \cdot E$, where $E$ is the combination of CX gates.
We first prove that we can move second column to another other column except the first column. We prove it by induction.

When qubit number $n = 2$. First consider moving column $2 \to 3$. We can do it by using $E = \text{CX}(0, 1)\text{CX}(1, 0)$. To move column $2 \to 4$, we can do it by using $E = \text{CX}(0,1)$.

Assume this proposition holds for $n = k$. Then for $n = k+1$, if we want to move  column $2\to i$, $i < 2^{k}$, then by hypothesis, we can find $E$ to do it (ignoring the $k$-th qubit). If we want to move column $2 \to i$, where $i > 2^{k}$, then we first apply $CX(0, k)$, then second column will move to the $2^k+1$ column. Then, by induction hypothesis, we can move $2^k+1$ column to any $2^k+j$ columns where $j \ge 1$. \\
Finally we need to show how to move column 2 to column $2^k$. We can do it with only two CX gates: $\text{CX}(0, k)\text{CX}(k, 0)$. \\
Thus this proposition holds for $n = k + 1$.

To prove CXs cannot implement arbitrary column swap, we just need to give a counterexample. Suppose we only want to swap the $(2^{n-1}+1)$-th column and the $(2^{n-1}+2)$-th column. If $n\ge 3$, this example swap is actually a controlled-swap operation which cannot be synthesized solely by CX gates.
\end{proof}

\subsection{Theorem~\ref{theo:rxrzcxefficiency}: Effectiveness of RX-RZ-CX ansatz}\label{sec:rx-rz-cx-eff}

Now, we will discuss the effectiveness of RX-RZ-CX ansatz together with several similar ansatzes.

\textbf{Theorem 3.1} (Efficiency of RX-RZ-CX ansatz with alternating entanglement)\textbf{.} 
\textit{For an $n$-qubit $2L$-layer RX-RZ-CX ansatz with alternating entanglement, the numbers of effective parameters (w.r.t parameter combination) $w_{xz}$ satisfies
$w_{xz} \ge (4n-3)*L$.}
\begin{proof}
Firstly, we consider the parameter combination of RX-RZ-CX ansatz with linear entanglement. We first consider the parameter combination of RX gates. In the RX-RZ-CX ansatz, we make the angles of RZ gate $\lambda_i = \frac{\pi}{2}$. Then, we have $\mathrm{RX}(\theta_i)\mathrm{RZ}(\frac{\pi}{2}) = \frac{\sqrt{2}}{2}\mathrm{RX}(\theta_i) + \frac{\sqrt{2}}{2}\mathrm{RX}(\theta_i)\mathrm{RZ}(\pi) = \frac{\sqrt{2}}{2}\mathrm{RX}(\theta_i) + \frac{\sqrt{2}}{2}\mathrm{RZ}(\pi)\mathrm{RX}(-\theta_i)$. In this case, $\forall i, j$, $\theta_i$ and $\theta_j$ cannot be combined because $\theta_{i} + \theta_{j}$ and $\theta_{i} - \theta_{j}$ will simultaneously appear in the resulted unitary matrix. Similarly, we have $\mathrm{RX}(\frac{\pi}{2})\mathrm{RZ}(\lambda_i) = \frac{\sqrt{2}}{2}\mathrm{RZ}(\lambda_i) + \frac{\sqrt{2}}{2}\mathrm{RX}(\pi)\mathrm{RZ}(\lambda_i) = \frac{\sqrt{2}}{2}\mathrm{RZ}(\lambda_i) + \frac{\sqrt{2}}{2}\mathrm{RZ}(-\lambda_i)\mathrm{RX}(\pi)$. Thus, by making angles of RX gate $\theta_i = \frac{\pi}{2}$, $\forall i, j$, $\lambda_i$ and $\lambda_j$ cannot be combined because $\lambda_{i} + \lambda_{j}$ and $\lambda_{i} - \lambda_{j}$ will simultaneously appear in the resulted unitary matrix. Thus, there is not parameter combination in RX-RZ-CX ansatz with linear entanglement.

For RX-RZ-CX ansatz with alternating entanglement, the parameter combination only happens on qubits without CX gates that single qubit gates can be combined into one u3 gate (generic single qubit rotation). As shown in Figure~\ref{fig:rxcxlayer}(c), for every 2 layers of RX-RZ-CX ansatz, at most 3 parameters are combined on qubits without CX gates. Thus, for a 2L-layer RX-RZ-CX ansatz with alternating entanglement, we have $w_{xz} \ge (4n-3)L$.
\end{proof}

\textbf{Corollary 3.2}\textbf{.} 
\textit{Theorem~\ref{theo:rxrzcxefficiency} can be generalized to other types of HEA such as RY-CX, RX-RY-CX and RY-RZ-CX. 
The number of effective parameters of an $n$-qubit $2L$-layer RY-CX, RX-RY-CX, or RY-RZ-CX ansatz with alternating entanglement is $w_{y} \ge (n-1)*2L$, $w_{yz} \ge (4n-3)*L$, or $w_{xy} \ge (4n-3)*L$, respectively.}

\begin{proof}
For RY-RZ-CX ansatz and RX-RY-CX ansatz, we can prove this proposition with the same technique as in Theorem~\ref{theo:rxrzcxefficiency} by using the following equations: \\
$\mathrm{RY}(\frac{\pi}{2})\mathrm{RZ}(\lambda_i) = \frac{\sqrt{2}}{2}\mathrm{RZ}(\lambda_i) + \frac{\sqrt{2}}{2}\mathrm{RY}(\pi)\mathrm{RZ}(\lambda_i) = \frac{\sqrt{2}}{2}\mathrm{RZ}(\lambda_i) + \frac{\sqrt{2}}{2}\mathrm{RZ}(-\lambda_i)\mathrm{RY}(\pi)$, \\
$\mathrm{RY}(\beta_i)\mathrm{RZ}(\frac{\pi}{2}) = \frac{\sqrt{2}}{2}\mathrm{RY}(\beta_i) + \frac{\sqrt{2}}{2}\mathrm{RY}(\beta_i)\mathrm{RZ}(\pi) = \frac{\sqrt{2}}{2}\mathrm{RY}(\beta_i) + \frac{\sqrt{2}}{2}\mathrm{RZ}(\pi)\mathrm{RY}(-\beta_i)$, \\
$\mathrm{RX}(\frac{\pi}{2})\mathrm{RY}(\beta_i) = \frac{\sqrt{2}}{2}\mathrm{RY}(\beta_i) + \frac{\sqrt{2}}{2}\mathrm{RX}(\pi)\mathrm{RY}(\beta_i) = \frac{\sqrt{2}}{2}\mathrm{RY}(\beta_i) + \frac{\sqrt{2}}{2}\mathrm{RY}(-\beta_i)\mathrm{RX}(\pi)$, \\
$\mathrm{RX}(\theta_i)\mathrm{RY}(\frac{\pi}{2}) = \frac{\sqrt{2}}{2}\mathrm{RX}(\theta_i) + \frac{\sqrt{2}}{2}\mathrm{RX}(\theta_i)\mathrm{RY}(\pi) = \frac{\sqrt{2}}{2}\mathrm{RX}(\theta_i) + \frac{\sqrt{2}}{2}\mathrm{RY}(\pi)\mathrm{RX}(-\theta_i)$.

For RY-CX ansatz, we prove it with the same technique as in Theorem~\ref{theo:rxrzcxefficiency} by using the fact that $\mathrm{RY}(\theta) = \mathrm{RX}(-\frac{\pi}{2})\mathrm{RZ}(\theta)\mathrm{RX}(\frac{\pi}{2}) = \mathrm{RZ}(\frac{\pi}{2})\mathrm{RX}(\theta)\mathrm{RZ}(-\frac{\pi}{2})$.
\end{proof}

\end{document}